\documentclass[11pt]{article}
\usepackage{amsmath,amsthm,amssymb,mathrsfs,amsfonts}
\usepackage[english]{babel}
\usepackage{times}
\usepackage{enumerate}
\usepackage{tikz}
\usetikzlibrary{arrows,positioning,fit,backgrounds}
\usetikzlibrary{calc}

\usepackage{array}
\newcolumntype{L}[1]{>{\raggedright\let\newline\\\arraybackslash\hspace{0pt}}m{#1}}
\newcolumntype{C}[1]{>{\centering\let\newline\\\arraybackslash\hspace{0pt}}m{#1}}
\newcolumntype{R}[1]{>{\raggedleft\let\newline\\\arraybackslash\hspace{0pt}}m{#1}}

\usepackage[final,hyperfootnotes=false]{hyperref}
\usepackage{xcolor}
\hypersetup{
    colorlinks,
    linkcolor={blue!80!black},
    citecolor={blue!80!black},
    urlcolor={blue!80!black}
}

\def\titlerunning#1{\gdef\titrun{#1}}
\makeatletter
\def\author#1{\gdef\autrun{\def\and{\unskip, }#1}\gdef\@author{#1}}
\def\address#1{{\def\and{\\\hspace*{18pt}}\renewcommand{\thefootnote}{}%
\footnote {#1}}%
\markboth{\autrun}{\titrun}}
\makeatother
\def\email#1{e-mail: #1}
\def\subjclass#1{{\renewcommand{\thefootnote}{}%
\footnote{\emph{Mathematics Subject Classification (2010):} #1}}}
\def\keywords#1{\par\medskip
\noindent\textbf{Keywords.} #1}

\newtheorem{thm}{Theorem}[section]
\newtheorem{cor}[thm]{Corollary}

\newtheorem{lem}[thm]{Lemma}
\newtheorem{prop}[thm]{Proposition}

\theoremstyle{definition}
\newtheorem{defn}[thm]{Definition}

\newtheorem{rem}[thm]{Remark}

\numberwithin{equation}{section}

\newcommand{\mb}{\mathbf}

\newcommand{\edt}[1]{#1}  

\DeclareMathOperator*{\rmd}{d}
\DeclareMathOperator*{\rmds}{d^{\star}}

\usepackage{pifont}

\begin{document}

\title{\vskip-1cm%
\Large\bf 
Second-Order Converses via Reverse Hypercontractivity\thanks{%
This work was supported in part by
NSF Grants CCF-1016625, CCF-0939370, DMS-1148711, by
ARO Grants W911NF-15-1-0479, W911NF-14-1-0094, and by the
Center for Science of Information. The main results of this paper were
initially announced in the conference paper \cite{lvv2017}, and form a 
part of the first author's Ph.D.\ thesis at Princeton University
\cite{liuthesis}.}}

\titlerunning{}

\author{Jingbo Liu \and Ramon van Handel \and Sergio Verd\'u}

\date{\vspace{-1.6cm}}

\maketitle

\address{Jingbo Liu, Institute for Data, Systems, and Society, 
Massachusetts Institute of Technology, 77 Massachusetts Avenue, Bldg.\ 
E18-421, Cambridge, MA 02139, USA; \email{jingbo@mit.edu}
\and
Ramon van Handel, Fine Hall 207, Princeton University, Princeton, NJ
08544, USA; \email{rvan@princeton.edu}
\and
Sergio Verd\'u:
\email{verdu@informationtheory.org}}


\begin{abstract} {\small
\edt{A strong converse shows that no procedure can beat the asymptotic (as 
blocklength $n\to\infty$) fundamental limit of a given 
information-theoretic problem for any fixed error probability. A 
second-order converse strengthens this conclusion by showing that the 
asymptotic fundamental limit cannot be exceeded by more than 
$O(\tfrac{1}{\sqrt{n}})$. While strong converses are achieved in a broad 
range of information-theoretic problems by virtue of the ``blowing-up 
method''---a powerful methodology due to Ahlswede, G\'acs and K\"orner 
(1976) based on concentration of measure---this method is fundamentally 
unable to attain second-order converses and is restricted 
to finite-alphabet settings.} Capitalizing on reverse hypercontractivity 
of Markov semigroups and functional inequalities, this paper develops the 
``smoothing-out'' method, an alternative to the blowing-up approach that
does not rely on finite alphabets and that leads to second-order 
converses in a variety of information-theoretic problems that were out of 
reach of previous methods.}
\end{abstract}

\vfill

\setcounter{tocdepth}{2}
\tableofcontents

\section{Introduction}\label{sec:1}

\subsection{Overview}

What are the fundamental limits of a given data science problem? The 
investigation of such questions typically follows a two-sided analysis. In 
the \emph{achievability} part, one shows the existence of a procedure 
achieving a certain performance, e.g., error probability. In the 
\emph{converse} part, one shows that no procedure can accomplish better 
performance than a certain lower bound. This basic structure is common to 
a wide range of problems that span information theory, statistics, and 
computer science. While the study of achievability generally depends 
strongly on the special features of the problem at hand, many converse 
bounds that arise in different areas rely essentially on 
information-theoretic methods.

Frequently, the upper and lower bounds on optimal performance provided by 
the achievability and converse analyses do not coincide. Starting with 
Shannon (1948) \cite{CES48}, the emphasis in information theory has been 
on the analysis of the fundamental limits in the \emph{asymptotic} regime 
of long blocklengths. In that regime, the gap between many existing 
achievability and converse bounds does vanish and sharp answers can be 
found to questions such as the maximal transmission rate over noisy 
channels and the minimal data compression rate subject to a fidelity 
constraint. The last decade has witnessed a number of 
information-theoretic results (e.g., 
\cite{polyanskiy2010channel,YKSV,kostina2012fixed}), which are applicable 
in the \emph{non-asymptotic} regime. These results are strongly motivated 
by the need to understand the relevance of asymptotic limits to practical 
systems that may be subject to severe delay constraints, or scenarios 
where the alphabet size or the number of users is large compared to the 
number of channels or sources 
\cite{polyanskiy2010channel,orlitsky2006,kelly2011,polyanskiy2017}.

The development of a non-asymptotic information theory has required new 
and improved methods for investigating fundamental limits. \edt{In 
principle, it is not clear that the quantities that determine the 
fundamental limits in the asymptotic regime can accurately describe the 
performance at blocklengths of interest in realistic applications (e.g., 
1000 bits). This has led to the development of nonasymptotic bounds 
that capture more sophisticated distributional information on the relevant 
information quantities. Unfortunately, however, such bounds are often 
difficult to compute. A less accurate but more tractable approach to 
understanding performance at smaller blocklengths is to focus attention on 
so-called \emph{second-order} analysis, originally pioneered by Wolfowitz 
\cite{JW:1E} and Strassen \cite{strassen:64b} in the 1960s and 
significantly refined in recent years. In such bounds, the fundamental 
limits in the first-order (linear in the blocklength) asymptotics are 
sharpened by investigating the deviation from this asymptotic behavior to 
second order (square-root of the blocklength). In particular, in those 
situations where the first- and second-order asymptotics can be 
established precisely, the resulting bounds have often proven to be quite 
accurate except for very short blocklengths.

While precise second-order results are available in various basic
information-theoretic problems, more complicated setups, particularly 
those arising in multiuser information theory, have so far eluded 
second-order analysis. One of the challenges that emerged from this line 
of work is the development of second-order converses. Several existing 
approaches to obtaining second-order converses are briefly reviewed in 
Section \ref{sec:exist}; however, to date, a variety of 
information-theoretic problems have remained out of reach of such methods. 
On the other hand, the powerful general methodology introduced in 1976 by 
Ahlswede, G\'acs, and K\"orner \cite{ahlswede_bounds_cond1976} and 
exploited extensively in the classical book \cite{csiszar2011information} 
(see also the recent survey \cite{raginsky2014concentration}) has proven 
instrumental for proving converses in network information theory. Although 
this widely used technique yields converse results in a broad range of 
problems almost as a black box, it is fundamentally unable to yield 
second-order converses and is restricted to finite-alphabet settings.}

Inspired by a result by Margulis \cite{margulis1974probabilistic}, the 
method of Ahlswede, G\'acs, and K\"orner is based on a remarkable 
application of the concentration of measure phenomenon on the Hamming cube 
\cite{ledoux,boucheron2013}, which is known in information theory as the 
``blowing-up lemma''. Historically, this is probably the very first 
application of modern measure concentration to a data science problem. 
One of the main messages of this paper is that, surprisingly, measure 
concentration turns out not to be the right approach after all in this 
original application. Instead, we will revisit the theory of Ahlswede, 
G\'acs, and K\"orner based not on the violent ``blowing-up'' operation, 
but on a new and more pacifist ``smoothing out'' principle that exploits 
reverse hypercontractivity of Markov semigroups and functional 
inequalities. With this gentler touch, we are able to eliminate the 
inefficiencies of the blowing-up method and obtain second-order converses, 
essentially for free, in many information-theoretic problems that were out 
of reach of previous methods.

\subsection{Weak, strong, and second-order converses}
\label{sec:exist}

As a concrete basis for discussion, let us consider the basic setup of 
single-user data transmission through noisy channels, in which there is a 
three-way tradeoff between code size, blocklength, and error probability. 
Suppose we wish to transmit an equiprobable message $W\in\{1,\ldots,M\}$ 
through a noisy channel with given blocklength $n$.\footnote{%
	A \emph{channel} is a sequence of random transformations $\{ 
	P_{Y^n|X^n } \}$ indexed by blocklength $n$.}
We encode each possible message using a codebook $c_1,\dots,c_M\in 
\mathcal{X}^n$. What is the largest possible size $M$ of the codebook that 
can be decoded with error probability (averaged over equiprobable 
codewords and channel randomness) at most $\epsilon$?
\edt{For memoryless channels and in various more general situations,}
the maximum code size $M^*(n,\epsilon)$ satisfies \cite{polyanskiy2010channel}
\begin{align}
	\ln M^*(n, \epsilon) =n C+{\rm
	Q}^{-1}(1-\epsilon)\sqrt{nV} +o_{\epsilon}(\sqrt{n}),
	\label{e_channel}
\end{align}
where \edt{the capacity $C$ and dispersion $V$ determine the
precise first- and second-order asymptotics,} and ${\rm Q}^{-1}(\cdot)$ is 
the inverse Gaussian tail probability function. 
For memoryless channels ($P_{Y^n|X^n} = P^{\otimes n}_{Y|X}$),
channel capacity and dispersion are given, respectively, by the quantities
\cite{CES48,polyanskiy2010channel,hayashi2006}
\begin{align}
C &= \max_{P_X} I (X;Y), \label{eq:capacity} \\
V &= \mathrm{Var} \left[ \imath_{X;Y} (X;Y) \right], \label{eq:dispersion}
\end{align}
where $\imath_{X;Y} (a;b) = \log \frac{\mathrm{d} P_{Y|X=a}}{\mathrm{d} 
P_Y}(b)$, $I(X;Y) = \mathbb{E} [ \imath_{X;Y} (X;Y)] $ is the mutual 
information, and \eqref{eq:dispersion} is evaluated for a $P_X$
that attains the maximum in \eqref{eq:capacity}.

To prove a result such as \eqref{e_channel}, we must address two separate 
questions. The achievability part (that is, the inequality $\ge$) requires 
us to show existence of a codebook $c_1,\ldots,c_M$ that attains the 
prescribed error probability. This is usually accomplished using the 
\emph{probabilistic method} due to Shannon \cite{CES48} which analyzes the 
error probability not of a particular code, but rather its average when 
the codebooks are randomly drawn from an auxiliary distribution. For many 
problems in information theory with known first-order asymptotics, an 
achievability bound with $\sim\sqrt{n}$ second-order term can be derived 
using random coding in conjunction with other techniques 
\cite{verduAllerton2012,yassaee2013technique,watanabe2015}.

In contrast, the converse part (that is, the inequality $\le$) claims that
\emph{no} code can exceed the size given in \eqref{e_channel} for the
given error probability $\epsilon$. The simplest and most widely used tool in converse analyses  is  Fano's inequality
\cite{fano1952}, which yields, in the memoryless case, the following
estimate:
\begin{align}
	\ln M^*(n, \epsilon) \le \frac{n}{1-\epsilon} C+ \frac{h(\epsilon)}{1- \epsilon}, \quad \epsilon \in (0,1).
\label{eq_weakconv}
\end{align}
Such a bound is called a \emph{weak converse}: it yields the correct
first-order asymptotics in the limit of vanishing error probability, namely,
\begin{align}
	\lim_{\epsilon \downarrow 0} \lim_{n \to \infty} \frac1n \ln M^*(n, \epsilon) \le C.
\label{eq_weakconv-2}
\end{align}
 However, \eqref{eq_weakconv} does not rule out that transmission rates exceeding $C$ 
 may be feasible for any given nonzero error probability $\epsilon$.
To that end we need a more powerful result known as the \emph{strong converse}, namely,
\begin{align}
	\ln M^*(n, \epsilon) \le nC + o_\epsilon(n), \quad \epsilon \in (0,1),
\label{e_strongchannel}
\end{align}
which can be proved by a variety of methods (e.g.,
\cite{JW59,shannon1967,wolfowitz1968,ahlswede_bounds_cond1976,verdu1994,polyanskiy2010channel,polyanskiy2010arimoto}).
Our interest in this paper is on the even stronger notion of a
\emph{second-order converse}
\begin{align}
	 \ln M^*(n, \epsilon) \le nC + O_\epsilon(\sqrt{n}),
	 \quad \epsilon \in (0,1),
\label{e_secondorderconverse}
\end{align}
which not only implies a strong converse \eqref{e_strongchannel} but 
yields the $\sqrt{n}$ behavior in \eqref{e_channel}.
\edt{We caution that second-order converses do not always yield the 
sharpest possible constant in the $\sqrt{n}$-term; however, they aim to 
capture at least qualitatively various features of the sharp second-order 
asymptotics illustrated by \eqref{e_channel}.}

Of course, the basic data transmission problem that we have discussed here 
for sake of illustration is particularly simple, as the exact second-order 
asymptotics \eqref{e_channel} are known. This is not the case in more 
complicated information-theoretic problems. In particular, there are many 
problems in multiuser information theory for which either only weak 
converses are known, or, at most, strong converses have been obtained 
using the blowing-up method. It is precisely in such situations that new 
powerful and robust methods for obtaining second-order converses are 
needed. Next, we
briefly discuss the three main approaches that have been developed in the
literature for addressing such problems.
\begin{enumerate}[1.]
%
\item The goal of the \emph{single-shot method} (e.g., 
\cite{CES57,verdu1994,polyanskiy2010channel,YKSV,kostina2012fixed,polyanskiy2010arimoto,CIT-086}) 
is to obtain non-asymptotic achievability and converse bounds without 
imposing any probabilistic structure on either sources or channels. 
Therefore, only the essential features of the problem come into play. 
Those non-asymptotic bounds are expressed not in terms of average 
quantities such as entropy or mutual information, but in terms of 
\emph{information spectra}, namely the distribution function of 
information densities such as $\imath_{X;Y} (X;Y)$. When coupled with the 
law of large numbers or the ergodic theorem and with central-limit 
theorem tools, the bounds become second-order tight. The non-asymptotic 
converse bounds for single-user data transmission boil down to the 
derivation of lower bounds on the error probability of Bayesian $M$-ary 
hypothesis testing followed by anonymization of the actual codebook. 
Often, those lower bounds are obtained by recourse to the analysis of an 
associated auxiliary binary hypothesis testing problem. This converse 
approach has been successfully applied to some problems of multiuser 
information theory such as Slepian-Wolf coding, multiple access 
channels \cite{han_information_spectrum_methods}, and broadcast channels 
\cite{oohama2015}. Its application to other network setups is however a 
work in progress. 
\item 
\edt{\emph{Type class analysis} has been used extensively since 
\cite{shannon1967} and was popularized by 
\cite{csiszar2011information} mainly in the context of error exponents; 
however, it applies also to second-order analysis (see, e.g., 
\cite{CIT-086}). The idea behind this method is that to obtain lower 
bounds, we may consider a situation where the decoder is artificially 
given access to the \emph{type} (empirical distribution) of 
the source or channel sequences. Conditioned on each type, the 
distribution is equiprobable on the type class, so the evaluation of the 
conditional error probability is reduced to a combinatorial problem (this 
has been referred to as the ``skeleton'' or ``combinatorial kernel'' of the 
information-theoretic problem \cite{ahlswede1979coloring}).
However, this combinatorial problem is not easily 
solved in side information problems (without additional ideas such as 
the blowing-up lemma). Moreover, by its nature, the method of types is 
restricted to finite alphabets and memoryless channels.}
\item The method using the \emph{blowing-up lemma (BUL)} of 
Ahlswede-G\'acs-K\"orner 
\cite{ahlswede_bounds_cond1976,csiszar2011information,raginsky2014concentration} 
uses a completely different idea to attain converse bounds: rather than 
try to reduce the given converse problem to a simpler one (e.g., to a 
binary hypothesis testing problem or a codebook that uses a single type), 
the BUL method is in essence a general technique for bootstrapping a 
strong converse from a weak converse. Even when the error probability 
$\epsilon$ is fixed, the concentration of measure phenomenon implies that 
all sequences except those in a set of vanishing probability differ in at 
most a fraction $o(1)$ of coordinates from a correctly decoded sequence. 
One can therefore effectively reduce the regime of fixed error probability 
to one of vanishing error probability, where a weak converse suffices, 
with negligible cost. The advantage of this method is that it is very 
broadly applicable. However, as will be discussed below, quantitative 
bounds obtained from this method are always suboptimal and second-order 
converses are fundamentally outside its reach. Moreover, the perturbation 
argument used in this approach is restricted to finite alphabets. 
\end{enumerate}
\edt{The single-shot and type class analysis 
methods yield second-order converses, but there are various problems in 
network information theory that have remained so far outside their reach. 
In contrast,} the BUL method has been successful in establishing strong 
converses for a wide range of problems, including all settings in 
\cite{csiszar2011information} with known single-letter rate region; see 
\cite[Ch.~16]{csiszar2011information}. For some problems in network 
information theory, such as source coding with compressed side information 
\cite{csiszar2011information}, BUL remained hitherto the only method for 
establishing a strong converse \cite[Section~9.2]{CIT-086}. However, the 
generality of the method comes at the cost of an inherent inefficiency, 
which prevents it from attaining second-order converses and prevents its 
application beyond the finite alphabet setting.

In this paper, we will show that one can have essentially the best of both
worlds: the inefficiency of the blowing-up method can be almost entirely
overcome by revisiting the foundation on which it is based. The resulting
theory provides a canonical approach for proving second-order asymptotic converses and is applicable to a wide
range of information-theoretic problems (including problems with general
alphabets) for which no such results were known.

\subsection{``Blowing up'' vs ``smoothing out''}

In order to describe the core ingredients of our approach, let us begin by
delving into the main elements of the blowing-up method of
Ahlswede, G\'acs and K\"orner (a detailed treatment in a toy example will
be given in Section \ref{sec_pre}).

The concentration of measure phenomenon is one of the most important ideas
in modern probability \cite{ledoux,boucheron2013}. It states that for many
high-dimensional probability measures, almost all points in the space are
within a small distance of any set of fixed probability. This basic
principle may be developed in different settings and has numerous
important consequences; for example, it implies that Lipschitz functions
on high-dimensional spaces are sharply concentrated around their median, a
fact that will not be used in the sequel (but is crucial in many other
contexts). The following modern incarnation of the concentration property
used in the work of Ahlswede-G\'acs-K\"orner is due to Marton
\cite{marton1986}; see
also~\cite[Lemma~3.6.2]{raginsky2014concentration}.
\begin{lem}[\bf Blowing-up lemma] \label{bluplem}
Denote the
$r$-\textit{blowup} of $\mathcal{A}\subseteq\mathcal{Y}^n$ by
\begin{align}
	\mathcal{A}_r := \{ v^n\in\mathcal{Y}^n:
	d_n(v^n,\mathcal{A})\le r\},\label{e_BlowUp}
\end{align}
where $d_n$ is the Hamming distance on $\mathcal{Y}^n$. Then
\begin{align}
	P^{\otimes n}[\mathcal{A}_r] \ge
	1-e^{-c^2}\quad\mbox{for}\quad
	r=\sqrt{\frac{n}{2}}\bigg(
	\sqrt{\ln\frac{1}{P^{\otimes n}[\mathcal{A}]}} + c\bigg),
	\label{e_BUL}
\end{align}
for any $c>0$ and any probability measure $P$ on $\mathcal{Y}$.
\end{lem}
\edt{For example,} if $n = 5 \times 10^9$ and $P^{\otimes n}[\mathcal{A}] 
= e^{-100}$, we can achieve $P^{\otimes n}[\mathcal{A}_r] \ge 1 - 
e^{-100}$ by letting $r = 10^6 = 0.0002 n$. Asymptotically, if $P^{\otimes 
n}[\mathcal{A}]$ does not vanish, then $P^{\otimes n}[\mathcal{A}_r] \to 
1$ as long as $r\gg\sqrt{n}$. In other words, the rather remarkable fact 
is that we can drastically increase the probability of a set by perturbing 
only a very small ($\approx n^{-1/2}$) fraction of coordinates of each of 
its elements.

Ahlswede, G\'acs and K\"orner realized how to leverage the BUL
 to prove strong converses. Suppose one is in a situation
where a weak converse,  such as \eqref{eq_weakconv}, can be proved through Fano's inequality
or any other approach. This can
be done in  all information-theoretic problems with known
first-order asymptotics. However, a weak converse only yields the
correct first-order constant when the error probability $\epsilon$ is allowed to vanish,
while we are interested in the regime of constant $\epsilon$. Let
$\mathcal{A}$ be the set of correctly decoded sequences, whose probability
is $1-\epsilon$. By the blowing-up lemma, a very slight blow-up
$\mathcal{A}_r$ of this set will already have probability $1-o(1)$. We now
apply the weak converse argument using $\mathcal{A}_r$ instead of
$\mathcal{A}$. On the one hand, this provides the desired first-order
term in \eqref{eq_weakconv}, as $1-\epsilon$ is replaced by $1-o(1)$. On
the other hand, we must pay a price in the argument for replacing the true
decoding set $\mathcal{A}$ by its blowup $\mathcal{A}_r$. If
$r=o(n)$, the latter turns out to contribute only to lower order and
thus a strong converse is obtained.

The beauty of this approach is that it provides a very general recipe for
upgrading a weak converse to a strong converse, and is therefore widely
applicable. However, the method has (at least) two significant drawbacks:
\begin{enumerate}[$\bullet$]
\item It is designed to yield a strong converse, not the stronger
second-order asymptotic converse. Therefore, it is not surprising
that it fails to yield  second-order behavior
that we expect from
\eqref{e_channel}: when optimized, the BUL method appears unable to give
a bound better than $O(\sqrt{n}\log^{\frac{3}{2}}n)$
(e.g., \cite[Thm.~3.6.7]{raginsky2014concentration}). This is already
suggested by Lemma \ref{bluplem} itself: to obtain a $\sim\sqrt{n}$
second-order term from \eqref{eq_weakconv}, we would need $\mathcal{A}_r$
to have probability at least $1-O(n^{-1/2})$. That would require
perturbing at least $r\sim\sqrt{n\log n}$ coordinates, which already gives
rise to additional logarithmic factors. Thus, the blowing-up operation is
too crude to recover the correct second-order behavior. In Appendix
\ref{app_data}, we will show that this is not an inefficiency in the
blowing-up lemma itself, but is in fact an insurmountable problem of
any method that is based on set enlargement.
\item The argument relies essentially on the finite-alphabet setting. This
is not because of the blowing-up lemma, which works for any alphabet
$\mathcal{Y}$, but because we must control the price paid for replacing
$\mathcal{A}$ by $\mathcal{A}_r$.
While Lemma \ref{bluplem} gives a lower bound on $P^{\otimes n}[\mathcal{A}_r]$
as a function of $P^{\otimes n}[\mathcal{A}]$,
we can also upper bound $P^{\otimes n}[\mathcal{A}_r]$
as a function of $P^{\otimes n}[\mathcal{A}]$  by the
following simple argument, which relies crucially on the finiteness of the alphabet.
\end{enumerate}
\begin{lem} \label{lemma:fin}
Suppose that $|\mathcal{Y}| < \infty$ and that  $P(a) >0 $ for all $a \in \mathcal{Y}$. Then
\begin{equation}  \label{eq:lemma:fin}
r \ln \frac{r}{n e K} \leq  \ln \frac{P^{\otimes n}[\mathcal{A}]}{P^{\otimes n}[\mathcal{A}_r]} \leq 0.
\end{equation}
where $K = \frac{|\mathcal{Y}|}{\min_{a\in \mathcal{Y}} P(a)}$. Therefore, if $r = o(n)$, then
\begin{equation}  \label{eq:lemma:fin:2}
\lim_{n \to \infty} \frac1n  \ln \frac{P^{\otimes n}[\mathcal{A}]}{P^{\otimes n}[\mathcal{A}_r]} = 0.
\end{equation}
\end{lem}
\begin{proof}
The right inequality in  \eqref{eq:lemma:fin} follows from $\mathcal{A} \subset \mathcal{A}_r$.
The set $\mathcal{A}_r$ is the overlapping union of spheres centered at the elements of $\mathcal{A}$,
each of which contains fewer than $\binom{n}{r} |\mathcal{Y}|^r$ elements. This fact, along with the  crude bound
\begin{equation}
\frac{P^{\otimes n} (y^n)}{P^{\otimes n} (z^n)} \geq 
\left(
\min_{a\in \mathcal{Y}} P(a)
\right)^{d_n ( y^n, z^n) }
\end{equation}
yields
\begin{equation}
P^{\otimes n}[\mathcal{A}_r] \leq \binom{n}{r}K^r  P^{\otimes n}[\mathcal{A}].
\end{equation}
Then, the left inequality in \eqref{eq:lemma:fin} follows from $ \binom{n}{r} \leq  \left( \frac{e n}{r}\right)^r $.
\end{proof}

The main contribution of this paper is to show that the shortcomings of
the blowing-up method can be essentially eliminated while retaining its
wide applicability. This is enabled by two key ideas that play
a central role in our theory.
\begin{enumerate}[1.]
\item \emph{Functional inequalities}.
To prove a weak converse such as \eqref{eq_weakconv}, one must relate the
relevant information-theoretic quantity (e.g., mutual information) to the
error probability (i.e., the probability of the decoding sets). This
connection is generally made using a data processing inequality. However,
in BUL-type methods, we no longer  work directly with the
original decoding sets, but rather with a perturbation of these that has
better properties. Thus, there is also no reason to restrict attention to
\emph{sets}: we can replace the decoding set by an arbitrary
\emph{function}, and then
control the relevant information functionals through their variational
characterization (that is, by convex duality). Unlike the data processing
inequality, such variational characterizations are in principle sharp and
provide a lot more freedom in how to approximate the decoding set.
\item \emph{``Smoothing out'' vs.\ ``blowing up''}. Once one makes the
psychological step of working with functional inequalities, it becomes
readily apparent that the idea of ``blowing up'' the decoding sets is much
more aggressive than necessary to obtain a strong converse. What turns out
to matter is not the overall size of the decoding set, but only the
presence of very small values of the function  used in the
variational principle. Modifying the indicator function of a set to
eliminate its small values can be accomplished by a much smaller
perturbation than is needed to drastically increase its probability: this
basic insight explains the fundamental inefficiency of the classical BUL
approach.

To implement this idea, we must identify an efficient method to improve the
positivity of an indicator function. To this end, rather than ``blowing up''
the set by adding all points within Hamming distance $r$, we will ``smooth
out'' the indicator function by averaging it locally over points at distance
$\sim r$. More precisely, this averaging will be performed by means of a
suitable Markov semigroup, which enables us to apply the \emph{reverse
hypercontractivity} phenomenon 
\cite{borell1982positivity,mossel2013reverse} to establish strong
positivity-improving properties. Such hypercontractivity phenomenon replaces, in our approach,  the
much better known concentration of measure phenomenon that was exploited in
the BUL method. We will show, moreover, that Markov semigroup perturbations
are much easier to control than their blowing-up counterparts, so that our
method extends readily to general alphabets, Gaussian channels, and channels
with memory.
\end{enumerate}
When combined, the above ideas provide a powerful machinery for developing
\edt{second-order converses}. For example, in the basic data transmission
problem that we discussed at the beginning of Section \ref{sec:exist}, our
method yields
\begin{equation} \label{prokofiev}
	\ln M^* (n, \epsilon) \le n C
	+2\sqrt{\ln\frac{1}{1-\epsilon}}\sqrt{n(\alpha-1)}
	+\ln\frac{1}{1-\epsilon}
\end{equation}
(cf.\ Theorem \ref{thm_fano_maximal}), where $\alpha$ is a certain quantity
closely related to the dispersion. 
\edt{When compared to the exact second-order
asymptotics \eqref{e_channel}, we see that
the second-order term has the correct scaling not only in the 
blocklength $n \to \infty$, but also in the error probability $\epsilon 
\to 1$ (as 
$\mathrm{Q}^{-1}(1-\epsilon)\sim\sqrt{2\ln\frac{1}{1-\epsilon}}$ as 
$\epsilon\to 1$). However, our method does not recover the exact 
dispersion in \eqref{e_channel}; nor does it capture the fact that 
in the small error regime $\epsilon < \frac12$, the second-order term in 
\eqref{e_channel} is in fact \emph{negative} (the second-order term in 
\eqref{prokofiev} is always positive). While the latter features are not 
directly achievable by the methods of this paper, they can be addressed
by combining our methods with type class analysis \cite{jingbo18}. The 
details of such a refined analysis are beyond the scope of this paper.}

Let us remark that there are various connections between the notions
of (reverse) hypercontractivity and concentration of measure; see, e.g.,
\cite{bobkov2001} and \cite[p.\ 116]{ledoux1996isoperimetry} in the
continuous case and \cite{mossel2006,mossel2013reverse} in the discrete
case. However, our present application of reverse hypercontractivity is
different in spirit: we are not using it to achieve concentration but
only a much weaker effect, which is the key to the efficiency of our method.
The sharpness and broad applicability of our method suggests that this may
be the ``right'' incarnation of the pioneering ideas of
Ahlswede-G\'acs-K\"orner: one might argue that the blowing-up method 
succeeded in its aims, in essence, because it approximates the natural smoothing-out 
operation.

\subsection{Organization}

We have sketched the main ideas behind our approach in broad terms. To
describe the method in detail, it is essential to get our hands dirty. To
this end, we develop both the blowing-up and smoothing-out approaches in
Section~\ref{sec_pre} in the simplest possible toy example: that of binary
hypothesis testing. While this problem is amenable to a (completely
classical) direct analysis, we view it as the ideal pedagogical setting in
which to understand the main ideas behind our general theory. 

The remainder of the paper is devoted to implementing these ideas in 
increasingly nontrivial situations.

In Section~\ref{sec_fano}, we use our 
approach to strengthen Fano's inequality with an optimal $O(\sqrt{n})$ 
second-order term. We develop both the discrete and the Gaussian cases, 
and illustrate their utility in applications to broadcast channels and the 
output distribution of good channel codes.

In Section~\ref{sec_change}, we use our approach to strengthen a basic image
size characterization bound of \cite{ahlswede_bounds_cond1976}, obtaining a
$O(\sqrt{n})$ second-order term. The theory is developed once
again both in the discrete and the Gaussian cases. We illustrate the utility
of our results in applications to hypothesis testing under communication
constraints, and to source coding with compressed side information.

\edt{The paper concludes with two appendices. In Appendix \ref{app_data}, 
we show that no method based on set enlargement can achieve second-order 
converses. Thus the functional viewpoint of this paper is essential. 
Finally, Appendix \ref{app_proofs} contains proofs of some technical 
results used in section \ref{sec_change}.}

\subsection{Notation}

We end this section by collecting common notation that will be used 
throughout the paper. First, we record two conventions that will always be 
in force:
\begin{enumerate}[$\bullet$]
\item All information-theoretic quantities (such as entropy, relative 
entropy, mutual information, etc.) will be defined in base $e$.
\item In all variational formulas (such as \eqref{e_var}, \eqref{e80}, 
\eqref{eq_dstar}, \eqref{e53}, etc.)\ it is implicit in the notation that
we optimize only over functions or measures for which each term of the 
expression inside the supremum are finite.
\end{enumerate}
We use standard information-theoretic notations
for relative entropy $D(P\|Q)$, conditional relative entropy
$D(P_{X|U}\|Q_{X|U}|P_U)= D(P_{X|U}P_U\|Q_{X|U}P_U)$, mutual information 
$I(X;Y)$, entropy $H(X)$, and differential entropy $h(X)$. 

We denote by $\mathcal{H}_+(\mathcal{Y})$ the set of nonnegative Borel 
measurable functions on $\mathcal{Y}$, and by 
$\mathcal{H}_{[0,1]}(\mathcal{Y})$ the subset of 
$\mathcal{H}_+(\mathcal{Y})$ with range in $[0,1]$. For a measure $\nu$ 
and $f\in\mathcal{H}_+(\mathcal{Y})$, we write $\nu(f):=\int f\,d\nu$ and 
$\|f\|_p^p = \|f\|_{L^p(\nu)}^p = \int |f|^p\,d\nu$. We will frequently 
use $\|f\|_{L^0(\nu)} := \lim_{q\downarrow 0}\|f\|_{L^q(\nu)} = 
\exp(\nu(\ln f))$ for a probability measure $\nu$. The measure of a set is 
denoted as $\nu[\mathcal{A}]$, and the restriction of a measure to a set 
is denoted $\mu|_{\mathcal{C}}[\mathcal{A}] := 
\mu[\mathcal{A}\cap\mathcal{C}]$. A random transformation $Q_{Y|X}$, 
mapping measures on $\mathcal{X}$ to measures on $\mathcal{Y}$, is viewed 
as an operator mapping $\mathcal{H}_+(\mathcal{Y})$ to 
$\mathcal{H}_+(\mathcal{X})$ according to 
$Q_{Y|X}(f):=\mathbb{E}[f(Y)|X=\cdot]$ where $(X,Y)\sim Q_{XY}$. The 
notation $U-X-Y$ denotes that the random variables $U,X,Y$ form a Markov 
chain. The cardinality of $\mathcal{Y}$ is denoted $|\mathcal{Y}|$, and 
$\|x^n\|$ denotes the Euclidean norm of a vector $x^n\in\mathbb{R}^n$. 
Finally, $y_i$ denotes the $i$th element of a sequence or vector, while 
$y^i$ denotes the components up to the $i$th one $(y_j)_{j\le i}$.

\section{Prelude: Binary Hypothesis Testing}\label{sec_pre}

\subsection{Setup}
The most elementary setting in which the ideas of this paper can be
developed is the classical problem of binary hypothesis testing. It should
be emphasized that our theory does not prove anything new in this setting:
due to the Neyman-Pearson lemma, the exact form of the optimal tests is
known and thus the analysis is amenable to explicit computation (we will
revisit this point in Section \ref{sec_set}). Nonetheless, the
simplicity of this setting makes it the ideal toy example in which to
introduce and discuss the main ideas of this paper.

In the binary hypothesis testing problem, we consider two competing 
hypotheses: data is drawn from a probability distribution on 
$\mathcal{Y}$, which we know is either $P$ or $Q$. Our aim is to test, on 
the basis of a data sample, whether it was drawn from $P$ or $Q$. More 
precisely, a (possibly randomized) \emph{test} is defined by a function 
$f\in \mathcal{H}_{[0,1]}(\mathcal{Y})$: when a data sample 
$y\in\mathcal{Y}$ is observed, we decide hypothesis $P$ with probability 
$f(y)$, and decide hypothesis $Q$ otherwise. Thus, we must consider two 
error probabilities:
\begin{align}
\pi_{P|Q}:=Q(f)=  ~\mbox{the probability that}~P~ \mbox{is decided when}~Q~ \mbox{is true} \nonumber\\
\pi_{Q|P}:=1 - P(f) =  ~\mbox{the probability that}~Q~ \mbox{is decided when}~P~ \mbox{is true} \nonumber
\end{align}
We aim to investigate the fundamental
tradeoff between $\pi_{P|Q}$ and $\pi_{Q|P}$
in the case of product measures $P\leftarrow P^{\otimes n}$,
$Q\leftarrow Q^{\otimes n}$; in other words,  what
is the smallest error probability $\pi_{P^{\otimes n}|Q^{\otimes n}}$ that
may be achieved by a test that satisfies $\pi_{Q^{\otimes n}|P^{\otimes n}}\le \epsilon\in (0,1)$? In
this setting, the exact first-order and second-order asymptotics 
are due to Chernoff \cite{HC56} and Strassen \cite{strassen:64b}, respectively, resulting in 
\begin{equation}
	\ln\frac1{\pi_{P^{\otimes n}|Q^{\otimes n}}} =
	nD(P\|Q)+\mathrm{Q}^{-1}(1-\epsilon)\sqrt{nV(P\|Q)} + 
	o_\epsilon(\sqrt{n}),
	\label{eq_strassen}
\end{equation}
where $V(P\|Q)=\mathrm{Var}_P(\ln\frac{dP}{dQ})$. 

The achievability ($\le$) part of \eqref{eq_strassen} is straightforward
(see Section \ref{sec_set}), so the main interest in the proof is
to obtain the converse ($\ge$). In this section, we will illustrate both
the blowing-up method of Ahlswede, G\'acs and K\"orner and the new
approach of this paper in the context of this simple problem, and compare
the resulting bounds to the exact second-order asymptotics
\eqref{eq_strassen}.

\subsection{The blowing-up method}\label{sec_bul}

For simplicity, to illustrate the blowing-up method in the context of the 
binary hypothesis testing problem, we restrict attention to 
\emph{deterministic} tests $f=1_\mathcal{A}$ for some 
$\mathcal{A}\subseteq\mathcal{Y}$ (that is, we decide hypothesis $P$ if 
$y\in\mathcal{A}$ and hypothesis $Q$ otherwise). This is not essential, 
but it simplifies the analysis.

As we mentioned, the blowing-up method is a general technique for 
upgrading weak converses to strong converses. In the present setting, a 
weak converse (for any $(P,Q)$, not necessarily product measures) follows 
in a completely elementary manner from the data processing property of 
relative entropy.

\begin{lem}[\bf Weak converse bound for binary hypothesis testing]
\label{lem_weakbht}
Let $P,Q$ be probability measures on $\mathcal{Y}$ and
$\mathcal{A}\subseteq\mathcal{Y}$ define the set of observations for which the deterministic test
decides $P$. 
If $\pi_{Q|P}=P[\mathcal{A}^c]\le\epsilon$,
then $\pi_{P|Q}=Q[\mathcal{A}]$ satisfies 
\begin{equation}
	\ln \frac1{\pi_{P|Q}} \le
	\frac{D(P\|Q)+ \ln 2}{1-\epsilon} .
\end{equation}
\end{lem}

\begin{proof}
By the data processing property of relative entropy, we have
\begin{align}
	 D ( P \| Q    ) &\geq
	P[\mathcal{A}]\ln\frac{P[\mathcal{A}]}{Q[\mathcal{A}]} +
	P[\mathcal{A}^c]\ln\frac{P[\mathcal{A}^c]}{Q[\mathcal{A}^c]}
	\label{uz}
	\\
	&\geq
	(1-\epsilon)\ln\frac{1}{Q[\mathcal{A}]}
	- h(P[\mathcal{A}^c]),
	\label{uzweak}
\end{align}
where $h(\epsilon) := -\epsilon\ln\epsilon - (1-\epsilon)\ln(1-\epsilon) \leq \ln 2$
is the binary entropy function.
\end{proof}

Specializing to product measures $P\leftarrow P^{\otimes n}$,
$Q\leftarrow Q^{\otimes n}$, Lemma \ref{lem_weakbht} yields
\begin{align}
	\ln\frac1{\pi_{P^{\otimes n}|Q^{\otimes n}}}
	\leq   \frac{nD ( P \| Q    ) + \ln 2}{1-\epsilon}
\label{n:uzweak}
\end{align}
for any test that satisfies $\pi_{Q^{\otimes n}|P^{\otimes n}}\le\epsilon$. 
While this is sufficient to conclude the weak converse [if 
$\pi_{Q^{\otimes n}|P^{\otimes n}} \to 0$, then $\pi_{P^{\otimes 
n}|Q^{\otimes n}}$ cannot vanish faster than $\exp ( - n D ( P \| Q ))$] 
it falls short of recovering the correct first-order asymptotics in the 
regime of fixed $\epsilon$, as we saw in \eqref{eq_strassen}.

\edt{The remarkable idea of Ahlswede, G\'acs and and K\"orner is that the 
argument of Lemma \ref{lem_weakbht} can be significantly improved by
applying} the data processing argument \eqref{uzweak} not to 
the test set $\mathcal{A}$ satisfying
\begin{equation}
\pi_{Q^{\otimes n}|P^{\otimes n}} =
P^{\otimes n}[\mathcal{A}^c]\le \epsilon,
\end{equation} 
but to its blow-up
$\mathcal{A}_r$.  Then
\begin{equation}
	\ln \frac1{Q^{\otimes n}[\mathcal{A}_r]}\le
	\frac{n D(P\|Q) + \ln 2}{P^{\otimes n}[\mathcal{A}_r]}.
	\label{eq_weakblup}
\end{equation}
We must now control both the gain and the loss caused by the blowup.
On the one hand, by the blowing-up Lemma \ref{bluplem},
$P^{\otimes n}[\mathcal{A}_r]=1-o(1)$ as long as $r\gg\sqrt{n}$, which
eliminates the $1-\epsilon$ factor in the weak converse \eqref{n:uzweak}.
On the other hand, assuming finite alphabets we can invoke Lemma \ref{lemma:fin} with $ r\ll n$
and $P\leftarrow Q$ (we may assume without loss of generality that $Q$ is 
positive on $\mathcal{Y}$) to obtain the strong converse
\begin{equation} \label{mozart}
	 \ln \frac1{\pi_{P^{\otimes n}|Q^{\otimes n}}} \le
	nD(P\|Q) + o_\epsilon(n).
\end{equation} 
With a little more effort, we can
optimize the argument over $r$ and quantify the magnitude of the 
lower-order term.

\begin{prop} \label{prop:blowBHT}
Assume $|\mathcal{Y}|<\infty$.
Any deterministic test between
$P^{\otimes n}$ and $Q^{\otimes n}$ on
$\mathcal{Y}^n$ such that
$\pi_{Q^{\otimes n}|P^{\otimes n}} \le \epsilon\in (0,1)$
satisfies
\begin{align}
	\ln\frac1{\pi_{P^{\otimes n}|Q^{\otimes n}}} \leq
	n D ( P \| Q    )  +
	 O(\sqrt{n} \log^{\frac32} n).
\label{n:uzstrong}
\end{align}
\end{prop}

\begin{proof}
By the blowing-up Lemma \ref{bluplem}, we have (since $3 > \sqrt{2} - \frac1{\sqrt{n}}$)
\begin{equation}
	P^{\otimes n}[\mathcal{A}_r] \ge
	1-e^{-r^2/n}\quad\mbox{for all }
	r\ge 3\sqrt{n\ln\frac{1}{1-\epsilon}}.
\label{eq_blupsimple}
\end{equation}
Assembling \eqref{eq:lemma:fin} (with $P \leftarrow Q$),  \eqref{eq_weakblup}  and \eqref{eq_blupsimple}
we obtain 
\begin{equation}
	\ln \frac1{\pi_{P^{\otimes n}|Q^{\otimes n}}}
	\le
	\frac{n D(P\|Q) + \ln 2}{1-e^{-r^2/n}}
	+ r \ln\frac{ K e \, n}{r}.
\end{equation}
Choosing $r\asymp\sqrt{n\log n}$ results in \eqref{n:uzstrong}.
\end{proof}

Re-examining the proof of Proposition \ref{prop:blowBHT}, we can easily 
verify that no other choice for the growth of $r$ with $n$ may accelerate 
the decay of the slack term in \eqref{n:uzstrong}. While the blowing-up 
method almost effortlessly turns a weak converse into a strong one, it 
evidently fails to result in a second-order converse. The rather crude 
bounds provided by Lemmas \ref{bluplem} and \ref{lemma:fin} may be 
expected to be the obvious culprits. It will shortly become evident, 
however, that the inefficiency of the method lies much deeper than 
expected: the major loss occurs already in the very first step \eqref{uz} 
where we apply the data processing inequality. We will in fact show in 
Appendix \ref{app_data} that any method based on the data processing 
inequality necessarily yields a slack term at least of order 
$\sim\sqrt{n\log n}$. To surmount this obstacle, we have no choice but to 
go back to the drawing board.

\subsection{The smoothing-out method}\label{sec_rhc}

The aim of this section is to introduce the key ingredients of the new
method proposed in this paper. As will be illustrated throughout this
paper, this method yields  second-order converses while retaining the
broad range of applicability of the blowing-up method. In the following,
there will be no reason to restrict attention to deterministic tests
$f=1_\mathcal{A}$ as in the previous section, so we will consider
arbitrary randomized tests $f\in\mathcal{H}_{[0,1]}$ from now on.

\subsubsection{Functional inequalities}

To prove a converse, we must relate the relevant 
information-theoretic quantity $D(P\|Q)$ to the properties of any given 
test $f$. This was accomplished above by means of the data processing 
inequality \eqref{uz}. However, as was indicated at the end of the 
previous section, this already precludes us from obtaining sharp 
quantitative bounds. The first idea behind our approach is to replace the 
data processing argument by a different lower bound: we will use 
throughout this paper \emph{functional inequalities} associated to 
information-theoretic quantities by convex duality. In the present 
setting, the relevant inequality follows from the Donsker-Varadhan 
variational principle for relative entropy \cite{donsker1977asymptotic} 
(see, e.g., \cite[(3.4.67)]{raginsky2014concentration})
\begin{align}
	D(P\|Q) = \sup_{g\in\mathcal{H}_+}
	\{ P(\ln g)-\ln Q(g)\}.
	\label{e_var}
\end{align}
Unlike the data
processing inequality, which can only attain equality in trivial
situations, the variational principle \eqref{e_var} always attains
its supremum by choosing $g\leftarrow\frac{dP}{dQ}$. Therefore,
unlike the data processing inequality, in principle an application of
\eqref{e_var} need not entail any loss.

What we must now show is how to choose the function $g$ in \eqref{e_var}
to capture the properties of a given test $f$. Tempting as it is, the
choice $g\leftarrow f$ is dismal: for example, in the case of
deterministic tests $f=1_\mathcal{A}$, generally $P(\ln
1_{\mathcal{A}})=-\infty$ and we do not even obtain a weak converse.
Instead, inspired by the blowing-up method, we may apply \eqref{e_var} to
a suitably chosen perturbation of $f$. Let us first develop the
argument abstractly so that we may gain insight into the requisite
properties. Suppose we can design a mapping
$T:\mathcal{H}_{[0,1]}\to\mathcal{H}_{[0,1]}$ (which plays the role 
of
the blowing-up operation in the present setting) that
satisfies:
\begin{enumerate}[1.]
\item For any test $f\in\mathcal{H}_{[0,1]}$ on $\mathcal{Y}^n$
with $P^{\otimes n}(f)\ge 1-\epsilon$, we have
\begin{equation}
	P^{\otimes n}(\ln Tf) \ge -o_\epsilon(n).
\label{eq_smdown}
\end{equation}
\item For any test $f\in\mathcal{H}_{[0,1]}$ on $\mathcal{Y}^n$, we have
\begin{equation}
	\ln Q^{\otimes n}(Tf) \le
	\ln Q^{\otimes n}(f) + o(n).
\label{eq_smup}
\end{equation}
\end{enumerate}
Setting $g\leftarrow Tf$ in \eqref{e_var} and using
\eqref{eq_smdown} and \eqref{eq_smup}, we immediately deduce a strong
converse: for any test $f\in\mathcal{H}_{[0,1]}$ such that
$\pi_{Q^{\otimes n}|P^{\otimes n}} := P^{\otimes n}(1-f) \le \epsilon$,
the error probability $\pi_{P^{\otimes n}|Q^{\otimes n}}:=Q^{\otimes
n}(f)$ satisfies \eqref{mozart}.
Besides replacing the data processing inequality by the variational
principle, the above logic parallels the blowing-up method:
\eqref{eq_smdown} plays the role of the blowing-up Lemma \ref{bluplem},
while \eqref{eq_smup} plays the role of the counting estimate
\eqref{eq:lemma:fin:2}.

Nonetheless, this apparently minor change of perspective lies at the heart
of our theory. To explain why it provides a crucial improvement, let us
pinpoint the origin of the inefficiency of the blowing-up method. The
purpose of the blowing-up operation is to increase the probability of a
test: given $P^{\otimes n}(f)\ge 1-\epsilon$, one designs a blow-up
$f\mapsto\tilde f$ so that $P^{\otimes n}(\tilde f)=1-o(1)$. However, when
we use the the sharp functional inequality \eqref{e_var} rather than the
data processing inequality, we do not need to control $P^{\otimes
n}(\tilde f)$, but rather $P^{\otimes n}(\ln\tilde f)$. The latter is
dominated by the \emph{small values} of $f$, not by its overall magnitude.
Therefore, an efficient perturbation of $f$ should not seek to blow it up
but only to boost its small values, which may be accomplished at a much
smaller cost than the blowing-up operation. It is precisely this insight
that will allow us to eliminate the inefficiency of the blowing-up method
and attain sharp second-order bounds.

In order to take full advantage of this insight, we must understand how to
design efficient perturbations $f\mapsto Tf$. The second key ingredient of
our method is its main workhorse:  a general mechanism to
implement \eqref{eq_smdown} and \eqref{eq_smup} so that their speed of decay
will be such that the slack term in \eqref{mozart} is in fact $O \left( n^{-1/2} \right).$

\subsubsection{Simple semigroups}
The essential intuition that arises from the above discussion is that in
order to obtain efficient bounds in \eqref{eq_smdown} and \eqref{eq_smup},
we must design an operation $f\mapsto Tf$ that is
\emph{positivity-improving}: it boosts the small values of $f$
sufficiently to ensure that $P^{\otimes n}(\ln Tf)$ is not too small. 
In this subsection we design a suitable transformation $T$, and in 
Section \ref{sec:mossel} we show that it achieves the desired goal.

Let $\mathcal{Y}$ be an arbitrary alphabet and let $P$ be any probability
measure thereon.
We say $(T_t)_{t\ge 0}$ is a \emph{simple semigroup}\footnote{
	Readers who are unfamiliar with semigroups may ignore this
	terminology; while the semigroup property plays an important role
	in the proof of Theorem \ref{thm_rhc}, it is not used
	directly in this paper.
}
with stationary measure $P$ if
\begin{align}
	T_t\colon
	\mathcal{H}_+(\mathcal{Y})\to \mathcal{H}_+(\mathcal{Y}),\qquad
	f\mapsto e^{-t}f+(1-e^{-t})P(f).
\label{e_simple}
\end{align}
In the i.i.d.~case $\mathcal{Y}\leftarrow\mathcal{Y}^n$, $P\leftarrow
P^{\otimes n}$ we consider their tensor product
\begin{align}
	T_t:=[e^{-t}+(1-e^{-t})P]^{\otimes n}.
\label{e_tens}
\end{align}
\edt{We will use $T=T_t$, for a suitable choice of $t$, as a
positivity-improving operation.}

It is instructive to examine the effect of the operator in \eqref{e_tens} 
on indicator functions. For that purpose, we introduce the following 
ad-hoc notation: if $(v^n, w^n) \in \mathcal{Y}^n \times \mathcal{Y}^n$ 
and $I \subset \{1, \ldots , n\}$, then $v^n I w^n \in \mathcal{Y}^n$ 
is defined by
\begin{align}
\left( v^n I w^n \right)_i =
\left\{
\begin{array}{ll}
v_i ,& i \in I \\
w_i ,& i \in I^c .
\end{array}
\right.
\end{align}
Then, using 
\begin{align}
(aQ+(1-a)P)^{\otimes n}=
\sum_{I\subset \{1, \ldots , n\}}a^{n-|I|}(1-a)^{|I|}P^{\otimes I}Q^{\otimes I^c},
\end{align}
the application of the operator in \eqref{e_tens} to the indicator function becomes
\begin{align}
	T_t1_{\mathcal{A}}(y^n) =
	\mathbb{E}[1_{\mathcal{A}}(
	Z^n \mathbf{I} y^n)],
\end{align}
where $I$ is a random subset
of $ \{1, \ldots , n\}$ obtained by including each element independently
with probability
$1-e^{-t}$ (in particular, $|\mathbf{I}|\sim\mathrm{Binom}(n,1-e^{-t})$),
and $Z^n\sim P^{\otimes n}$ is independent of $\mathbf{I}$.
In contrast, the
blowing-up operation may be expressed in terms of indicator functions as
\begin{align}
	1_{\mathcal{A}_r}(y^n) =
	\max_{|I|\le r}
	\max_{z^n\in\mathcal{Y}^n}
	1_{\mathcal{A}}( z^n I y^n).
\end{align}
From this perspective, we see that the semigroup operation is a
\emph{smoothing out} counterpart of the blowing-up operation:
while the blowing-up operation \emph{maximizes} the function over a local
neighborhood of size $r$, the semigroup operation \emph{averages} the
function over a random neighborhood of size $r\approx n(1-e^{-t})$.
What we will gain from smoothing is that it increases the small values
of $f$ (it is positivity-improving) without increasing the total mass
$P^{\otimes n}(T_tf)=P^{\otimes n}(f)$, so that the mass under
$Q^{\otimes n}$ cannot grow too much. In contrast, blowing-up
is designed to increase the mass $P^{\otimes n}(f)$; but
then the mass under $Q^{\otimes n}$ becomes large as well, which yields
the suboptimal rate achieved by the blowing-up method.

\subsubsection{Reverse hypercontractivity} \label{sec:mossel}
It is intuitively clear that $T_t$ is positivity improving: it maps
any nonnegative function to a strictly positive function. But the goal of
lower bounding $P^{\otimes n}(\ln Tf)$ is more ambitious.
This 
idea already appears in the probability theory literature in a very
different context: it was realized long ago by Borell
\cite{borell1982positivity} that Markov semigroups possess very strong
positivity-improving properties, which are described quantitatively by a
reverse form of the classical hypercontractivity phenomenon. While Borell
was motivated by applications in quantum field theory, we will show in
this paper that reverse hypercontractivity provides a powerful mechanism
that appears almost tailor-made for our present purposes. 

We will presently describe an important generalization of Borell's ideas 
to general alphabets due to Mossel et al.\ \cite{mossel2013reverse}, and 
show how it may be combined with the above ideas to obtain sharp 
non-asymptotic converses.

\begin{thm}[\bf Reverse hypercontractivity] \cite{mossel2013reverse}.
\label{thm_rhc}
Let $(T_t)_{t\ge 0}$ be a simple semigroup \eqref{e_simple} or an
arbitrary tensor product of
simple semigroups. Then
\begin{equation}
	\|T_tf\|_{L^q} \ge \|f\|_{L^p} \label{e_RHC}
\end{equation}
for any $0<q<p<1$, $f\in\mathcal{H}_+$, and $t\ge \ln\frac{1-q}{1-p}$.
In particular, letting $q \downarrow 0$, we have
\begin{equation}
	P(\ln T_t f) \ge \log \|f\|_{L^p (P)}. \label{e_RHC:q0}
\end{equation}
\end{thm}

An estimate of the form \eqref{eq_smdown} is almost immediate from 
\eqref{e_RHC:q0}. 
On the other hand, the estimate \eqref{eq_smup} will 
now follow from a simple change of measure argument. The following result 
combines these ingredients to derive a second-order converse for binary 
hypothesis testing.

\begin{thm} \label{thm:hyperBHT}
Any test between $P^{\otimes n}$ and $Q^{\otimes n}$ such that
$\pi_{Q^{\otimes n}|P^{\otimes n}} \le \epsilon$ satisfies
\begin{align}
	\ln\frac1{\pi_{P^{\otimes n}|Q^{\otimes n}}} \leq
	 n D ( P \| Q)
	+ 2\sqrt{\ln\frac{1}{1-\epsilon}}
	\sqrt{n(\alpha-1)}
	+\ln\frac{1}{1-\epsilon},
\label{n:uzverystrong}
\end{align}
where $\alpha = \big\|\frac{dP}{dQ}\big\|_\infty\ge 1$.
\end{thm}

\begin{proof}
We establish \eqref{eq_smdown} and \eqref{eq_smup} by
choosing $T=T_t$ as defined by \eqref{e_tens}. We fix any $t>0$ initialy
and optimize at the end of the proof.

Fix any test $f\in\mathcal{H}_{[0,1]}(\mathcal{Y}^n)$.
To establish \eqref{eq_smdown}, note that by \eqref{e_RHC:q0}
\begin{align}
	P^{\otimes n}(\ln T_t f)
	&\ge\ln\|f\|_{L^{1-e^{-t}}(P^{\otimes n})}
	\\
	&\ge \frac{1}{1-e^{-t}}\ln P^{\otimes n}(f)
	\label{e_14}
	\\
	&\ge \left(\frac{1}{t}+1\right)\ln P^{\otimes n}(f),\label{e_15}
\end{align}
where \eqref{e_14} used that $f\in[0,1]$ and
\eqref{e_15} follows from $e^t\ge 1+t$.

On the other hand, to establish \eqref{eq_smup}, we argue as follows:
\begin{align}
	Q^{\otimes n}(T_tf)
	&=Q^{\otimes n}((e^{-t}+(1-e^{-t})P)^{\otimes n}f)
\label{e17}
\\
	&=(e^{-t}Q+(1-e^{-t})P)^{\otimes n}f
\\
	&\le (e^{-t}+\alpha(1-e^{-t}))^nQ^{\otimes n}(f) 
\label{e_bdd}
\\
	&\le e^{(\alpha-1)nt}Q^{\otimes n}(f),\label{e18}
\end{align}
where $\alpha=\|\frac{dP}{dQ}\|_{\infty} \geq 1$, and
\eqref{e17} is just the definition of $T_t$.

Now assume that the test $f$ satisfies $\pi_{Q^{\otimes n}|P^{\otimes
n}}:=P^{\otimes n}(1-f)\le\epsilon$. Setting $P\leftarrow P^{\otimes n}$,
$Q\leftarrow Q^{\otimes n}$, and $g\leftarrow T_tf$ in
\eqref{e_var}, we obtain for all $t>0$
\begin{equation}\label{bach}
	 \ln \frac1{\pi_{P^{\otimes n}|Q^{\otimes n}}}
	\le
	nD(P\|Q)
	+ \bigg(\frac{1}{t}+1\bigg)\ln\frac{1}{1-\epsilon}
	+(\alpha-1)n t
\end{equation}
using \eqref{e_15} and \eqref{e18}. 
Minimizing \eqref{bach} with respect to $t$ yields \eqref{n:uzverystrong}.
\end{proof}

Beside resulting in a  second-order converse, the smoothing-out method
has an additional major advantage over the blowing-up method: the
change-of-measure argument \eqref{e18} is purely measure-theoretic in
nature and sidesteps the counting argument of Lemma \ref{lemma:fin}.
No analogue of the latter can hold
beyond the finite alphabet case: indeed, in general alphabets even the
blowup of a set of measure zero will have positive measure, ruling out any
estimate of the form \eqref{eq:lemma:fin:2}. In contrast, the result of Theorem \ref{thm:hyperBHT}
holds for any alphabet $\mathcal{Y}$ and requires only a bounded density
assumption. Even the latter assumption is not an essential restriction and
can be eliminated in specific situations by working with semigroups other
than \eqref{e_tens}. A particularly important example that will be
developed in detail later on in this paper is the case of Gaussian
measures (see, for example, Section~\ref{sec_gfano}).

The proof of Theorem \ref{thm:hyperBHT} illustrates the approach of this
paper in its simplest possible setting. However, the basic ideas that we
have introduced here form the basis of a general recipe that will be
applied repeatedly in the following sections to obtain second-order
 converses in a broad range of applications. The comparison
between the key ingredients of the blowing-up method and the smoothing-out
approach proposed in this paper is summarized in Table \ref{tab1}.
\begin{table}[t]
{\centering
\small
\begin{tabular}{|C{3.5cm}|C{3.5cm}|C{3.5cm}|}
  \hline 
  ~ & Blowing-up method & Smoothing-out method \\
  \hline
  Connecting information measures and observables
  &
  Data processing property \eqref{uz} & Convex duality \eqref{e_var} \\
  \hline
  Lower bound w.r.t.\ a given measure
  &  Concentration of measure (Lemma \ref{bluplem})
  &  Reverse hypercontractivity (Theorem \ref{thm_rhc}) \\
  \hline
  Upper bound w.r.t.\ reference measure
  & Counting argument (Lemma \ref{lemma:fin})
  & Change of measure \eqref{e_bdd} [or
  scaling argument \eqref{e60}] \\
  \hline
\end{tabular}
\caption{Main ingredients of the blowing-up and smoothing-out
approaches.}\label{tab1}
}
\end{table}

\subsubsection{Beyond product measures} 
Throughout this paper we focus for simplicity on stationary memoryless
systems, that is, those defined by product measures. However, our approach
is by no means limited to this setting. For the benefit of the
interested reader, let us briefly sketch in the context of Theorem
\ref{thm:hyperBHT} what modifications would be needed to adapt our
approach to general dependent measures. For an entirely different
application of our approach in a dependent setting, see \cite{jingbo18}.

Consider the problem of testing between two arbitrary (non-product)
hypotheses $P_n,Q_n$ on $\mathcal{Y}^n$. To adapt the proof of Theorem
\ref{thm:hyperBHT}, we need to introduce a hypercontractive semigroup with
stationary measure $P_n$. A natural candidate in this general setting is
the so-called \emph{Gibbs sampler} $T_tf(y^n)=\mathbb{E}_{y^n}[f(Y_t^n)]$,
where the Markov process $Y^n_t$ is defined by replacing each coordinate
with an independent draw from its conditional distribution
$P_{Y_i|Y_{\setminus i}}$ (where $Y_{\setminus i}:=(Y_j)_{j\ne i}$) at independent
exponentially distributed intervals. It was shown in
\cite{mossel2013reverse} that any semigroup that satisfies a modified
log-Sobolev inequality is reverse hypercontractive. In particular, such
inequalities may be established for the Gibbs sampler under rather general
weak dependence assumptions \cite{approximate_tens,marton15}.

On the other hand, we need to establish an upper bound on
$Q_n(T_tf)$. This may be done as follows. The Gibbs sampler satisfies
the differential equation \cite{approximate_tens}
\begin{equation}
	\frac{d}{dt} T_tf(y^n) = \sum_{i=1}^n
	\{P_{Y_i|Y_{\setminus i}=y_{\setminus i}}(T_tf)-T_tf(y^n)\}.
\end{equation}
We may therefore estimate for any $f\in\mathcal{H}_+(\mathcal{Y}^n)$
\begin{equation}
	\frac{d}{dt} Q_n(T_tf) =
	\sum_{i=1}^n
	\{Q_n(P_{Y_i|Y_{\setminus i}}(T_tf))-Q_n(T_tf)\}
	\le (\alpha-1)n\,Q_n(T_tf),
\end{equation}
where we used the tower property of conditional expectations and we defined
\begin{equation}
	\alpha :=
	\max_i\bigg\|\frac{P_{Y_i|Y_{\setminus i}}}{Q_{Y_i|Y_{\setminus i}}}\bigg\|_\infty.
\end{equation}
Solving the differential inequality yields precisely the same estimate
$Q_n(T_tf)\le e^{(\alpha-1)nt}Q_n(f)$ as was obtained  in
\eqref{e18} in the product case.

Putting together these estimates, we obtain for any test with
$\pi_{Q_n|P_n}\le\epsilon$ that
\begin{equation}
	\ln \frac1{\pi_{P_n|Q_n}} \leq
	n D ( P_n \| Q_n)
	+ \sqrt{C\ln\frac{1}{1-\epsilon}}
	\sqrt{n(\alpha-1)}
	+\ln\frac{1}{1-\epsilon},
\end{equation}
where $C$ is the modified log-Sobolev constant of $P_n$.
This extends our approach for the memoryless case to any dependent 
situation where a modified log-Sobolev inequality is available for $P_n$. 
A deeper investigation of dependent processes is beyond the scope of this 
paper.

\subsection{Achievability and optimality}
\label{sec_set}

Theorem \ref{thm:hyperBHT} gives a non-asymptotic converse bound for binary
hypothesis testing. To understand whether it is accurate, we also need
an upper bound (achievability). Such a bound was already stated in
\eqref{eq_strassen}. To motivate the following discussion, it is
instructive to give a quick proof of the achievability part of
\eqref{eq_strassen}.

\begin{lem}
There exist a sequence of binary hypothesis tests  with $\pi_{Q^{\otimes n}|P^{\otimes n}}
\le\epsilon$ such that
\begin{equation}
	\ln \frac1{\pi_{P^{\otimes n}|Q^{\otimes n}} } \ge
	nD(P\|Q) + \mathrm{Q}^{-1}(1-\epsilon)\sqrt{nV(P\|Q)}
	+ o_\epsilon(\sqrt{n}),
\label{eq_bhtach}
\end{equation}
provided that $V (P\|Q) :=\mathrm{Var}_P(\ln\frac{dP}{dQ})<\infty$.
\label{lem_bhcach}
\end{lem}

\begin{proof} 
For typographical convenience we will write
$D:=D(P\|Q)$ and $V:=V(P\|Q)$.
By the central limit theorem
\begin{equation}
	\lim_{n\to\infty}
	P^{\otimes n}\left[\imath_{P^{\otimes n}\|Q^{\otimes n}}
	\ge nD + \mathrm{Q}^{-1}(1-\epsilon)\sqrt{nV} \right]
	= 1-\epsilon,
\end{equation}
where we have defined the \emph{relative information}
\begin{align}
\imath_{P^{\otimes n}\|Q^{\otimes n}}(y^n):=
\ln\frac{dP^{\otimes n}}{dQ^{\otimes n}}(y^n)
=\sum_{i=1}^n\ln\frac{dP}{dQ}(y_i)
\end{align}
whose mean and variance with $y^n \sim P^{\otimes n}$ are $n D$ and $n V$, respectively.
We may therefore choose a deterministic
sequence $a_n=o_\epsilon(\sqrt{n})$ such that the deterministic test
$f_n=1_{\mathcal{A}_n}$ defined by
\begin{equation}
	\mathcal{A}=\{y^n\in\mathcal{Y}^n:\imath_{P^{\otimes
	n}\|Q^{\otimes n}}(y^n)
        \ge nD + \mathrm{Q}^{-1}(1-\epsilon)\sqrt{nV} - a_n\}
\label{eq_neyman}
\end{equation}
satisfies $\pi_{Q^{\otimes n}|P^{\otimes n}}\le \epsilon$ for all $n$. But
a simple Chernoff bound now yields
\begin{equation}
	\pi_{P^{\otimes n}|Q^{\otimes n}} \le
	e^{-nD - \mathrm{Q}^{-1}(1-\epsilon)\sqrt{nV} + a_n}
	Q^{\otimes n}(e^{\imath_{P^{\otimes n}\|Q^{\otimes n}}}),
\end{equation}
and the proof is completed by noting that $Q^{\otimes
n}(e^{\imath_{P^{\otimes n}\|Q^{\otimes n}}})=1$.
\end{proof}

Comparing the achievability bound \eqref{eq_bhtach} with our converse
\eqref{n:uzverystrong}, we see that the main features of the second-order
term are captured faithfully by the smoothing-out method, although it fails
to recover the precise constant in the second-order term:
$V(P\|Q)$ is replaced by its natural uniform
bound $V(P\|Q)\le
\big\|\frac{dP}{dQ}\big\|_\infty-1$ in our  converse.\footnote{%
	To see this, use $x\ln^2x\le(x-1)^2$ to show
	$\mathrm{Var}_P(\ln\frac{dP}{dQ})\le
	Q(\frac{dP}{dQ}\ln^2\frac{dP}{dQ})\le P(\frac{dP}{dQ}-1)$.}
\edt{Beside the optimal order scaling $\sim\sqrt{n}$, we recall that the bound 
behaves correctly as a function of $\epsilon$ (up to universal constant) 
for large error probabilities $\epsilon\to 1$; cf.\ the disucussion 
following \eqref{prokofiev}.}

As is illustrated by Lemma \ref{lem_bhcach}, the achievability analysis is 
conceptually simple in the binary hypothesis testing case thanks to the 
Neyman-Pearson lemma which identifies the optimal test. However, in 
information theory optimal procedures are very seldom known explicitly.  
Thus, the methodology we have introduced says nothing new about binary 
hypothesis testing. The point of the present method, however, is that it 
applies broadly in situations where such a direct analysis is far out of 
reach. In particular, in the general setting of Section \ref{sec_change} 
the present approach is currently the only known method to achieve sharp 
second-order converses in a variety of multiuser information theory 
problems.

\section{Second-Order Fano's Inequality}\label{sec_fano}

The aim of this section is to develop the smoothing-out methodology for
channel coding problems, of which a basic example was discussed in Section
\ref{sec:exist}.

Weak converses for channel coding problems can be obtained in great 
generality (cf.\ \cite{el1981proof}): this is the domain of Fano's 
inequality \cite{fano1952}, one of the most basic results in information 
theory, which gives an implicit upper bound on the error probability of an 
$M$-ary hypothesis testing problem.  For discrete memoryless channels, 
when combined with a list decoding argument, the blowing-up method 
strengthens Fano's inequality to a strong converse with 
$\sim\sqrt{n}\log^{\frac{3}{2}}n$ second-order term 
\cite{ahlswede_bounds_cond1976,marton1986,polyanskiy2014empirical, 
raginsky2014concentration}. In this section, we will show that the 
smoothing-out method results in a strong form of Fano's inequality that 
not only attains the optimal $\sim\sqrt{n}$ second-order term, but is 
applicable to a much broader class of channels. The power of this 
machinery will be illustrated in two typical applications.

Before we turn to the main results of this section, it is instructive to
give a short proof of a basic form of Fano's inequality.  
Although we state it and prove it for deterministic decoding, it also
holds for stochastic decoders.

\begin{lem}[\bf Fano's inequality]
\label{lem_fano}
Let $W\in\{1,\ldots,M\}$ be an equiprobable message to be transmitted over
a noisy channel $P_{Y|X}$. Let $c_1,\ldots,c_M\in\mathcal{X}$ be the
codewords corresponding to $W$, and let
$\mathcal{D}_1,\ldots,\mathcal{D}_M\subseteq\mathcal{Y}$ be the
disjoint decoding sets. Suppose the average probability of
correct decoding satisfies
\begin{equation}
	\frac{1}{M}\sum_{m=1}^M P_{Y|X=c_m}[\mathcal{D}_m] \ge 1-\epsilon.
\label{e_avgerr}
\end{equation}
Then
\begin{equation}
	\ln M \le \frac{I(W;Y)+\ln 2}{1-\epsilon},
\label{eq_fano}
\end{equation}
where $Y$ is the output of the channel $P_{Y|X}$ with input $X=c_W$.
\end{lem}

\begin{proof} 
Let $\hat W$ be the decoded message, that is, $\hat W=c_m$ when the output
$Y\in\mathcal{D}_m$.
The bound can be shown by reduction to an auxiliary binary hypothesis testing problem:
$P \leftarrow P_{W\hat W}$, $Q \leftarrow P_W\otimes P_{\hat W}$.
Then the conclusion follows by applying Lemma \ref{lem_weakbht} 
with $\mathcal{A}=\{(x,\hat x):x=\hat x\}$ since $I(W;\hat W)\le I(W;Y)$.
\end{proof}

The proof of Fano's inequality highlights the connection between the
channel coding problems investigated in this section and the simple
hypothesis testing problem of Section \ref{sec_pre}. In particular, it
suggests that the weak converse \eqref{eq_fano} may be strengthened to a
strong converse of the form $\ln M\le I(W;Y^n) + O(\sqrt{n})$ in the
setting of memoryless channels $P_{Y|X}\leftarrow
P_{Y^n|X^n}:=P_{Y|X}^{\otimes n}$. Unfortunately, the latter does not
follow from the strong converse obtained in Section \ref{sec_pre} for
binary hypothesis testing. Indeed,  
the framework of Section \ref{sec_pre} does not apply as
 the measures $P,Q$ that appear in the proof of Lemma
\ref{lem_fano} are not product measures even when the channel is
memoryless.
To sidestep this hurdle we will apply the
smoothing-out operation conditionally on $X^n$: that is, we will introduce
semigroups for which the channel $P_{Y^n|X^n=x^n}$ is the stationary
measure. The main additional challenge that arises is that the semigroup depends on the channel input
$x^n$ and must therefore be controlled uniformly in $x^n$.

\subsection{Bounded probability density case}\label{sec_fano_d}

In this section, we consider a random transformation $P_{X|Y}$ from
$\mathcal{X}$ to $\mathcal{Y}$ and denote by $P_{Y^n|X^n=x^n}:=
P_{Y|X=x_1}\otimes\cdots\otimes P_{Y|X=x_n}$ the corresponding $n$-fold memoryless
random transformation. The main assumption of this section is
that there exists a reference probability measure $\nu$ on $\mathcal{Y}$ 
such that
\begin{equation}
	\alpha :=
	\sup_{x\in\mathcal{X}}\left\|\frac{dP_{Y|X=x}}{d\nu}\right\|_{\infty}
	\in [1,\infty).
\label{e_assump}
\end{equation}
This assumption is automatically satisfied in the finite alphabet setting
$|\mathcal{Y}|<\infty$, in which case we may choose $\nu$ to be
equiprobable and consequently $\alpha\le|\mathcal{Y}|$. However, the
present setting is much more general: it applies to an arbitrary output
alphabet $\mathcal{Y}$ and requires only the existence of bounded
densities.

The main result of this section is the following strong form of Fano's
inequality (cf.\ Lemma \ref{lem_fano}) exhibiting a $\sqrt{n}$ second-order term..

\begin{thm}\label{thm_fano_maximal}
Assume \eqref{e_assump} holds. Let $W\in\{1,\ldots,M\}$ be an equiprobable
message, let $c_1,\dots,c_M\in\mathcal{X}^n$ be the codewords
corresponding to $W$, and let
$\mathcal{D}_1,\ldots,\mathcal{D}_M\subseteq\mathcal{Y}^n$ be
disjoint decoding sets.  Suppose that
\begin{align}
	\prod_{m=1}^M P_{Y^n|X^n=c_m}^{\frac{1}{M}}[\mathcal{D}_m]
	\ge 1-\epsilon.
	\label{e_geo}
\end{align}
Then
\begin{align}
	\ln M \le
	I(W;Y^n)
	+2\sqrt{\ln\frac{1}{1-\epsilon}}
	\sqrt{n(\alpha-1)}
	+\ln\frac{1}{1-\epsilon},
\label{e_fanodiscrete}
\end{align}
where $Y^n$ is the output of the memoryless channel $P_{Y^n|X^n} = P^{\otimes n}_{Y|X}$ with input
$X^n=c_W$.
\end{thm}

\begin{rem}
\label{rem_expur}
\edt{
The geometric average criterion \eqref{e_geo} is stronger than the 
average error criterion \eqref{e_avgerr} in Fano's inequality, but is 
weaker than the maximal error criterion $\min_m 
P_{Y^n|X^n=c_m}[\mathcal{D}_m]\ge 
1-\epsilon$. Both the average and maximal error criteria are commonly used 
in information theory, and our assumption
\eqref{e_geo} is intermediate between these two conventional criteria.

It is natural to ask whether the strong Fano inequality 
\eqref{e_fanodiscrete} remains valid even under the average error 
criterion, like the classical Fano inequality. This is not the case. In 
\cite{allerton_lcv2017}, a general notion of ``$\alpha$-decodability'' is 
introduced which subsumes the geometric average criterion, the average 
error criterion, and the maximum error criterion as the special cases 
$\alpha=0,1,-\infty$. It is shown there that $\alpha=0$ is the critical 
value for the existence of a strong Fano inequality. In this sense, the 
assumption \eqref{e_geo} of Theorem \ref{thm_fano_maximal} is essentially 
the best possible.

While the strong Fano inequality itself cannot hold under the average 
error criterion, it is possible in some applications of this inequality to 
upgrade the subsequent results to hold under the average error criterion 
by an additional argument known as \emph{codebook expurgation}. The 
idea behind this method is that if the average error criterion is met, 
then the maximal error criterion will be satisfied on a large subset of 
the codebook obtained by throwing out the worst half of the codewords 
(e.g., \cite[Theorem~7.7.1]{cover2012elements}).
By combining this device with 
Theorem~\ref{thm_fano_maximal} we can, for example, recover a strong 
converse under the average error criterion for the basic channel coding 
problem of Section \ref{sec:exist}. The expurgation argument cannot be 
applied directly to Theorem~\ref{thm_fano_maximal}, however, as $I(W;Y^n)$ 
may change significantly when we throw out codewords. Whether or 
not an expurgation argument is feasible depends on the manner in which 
Theorem~\ref{thm_fano_maximal} is used in a given application.}
\end{rem}

Before we turn to the proof of Theorem \ref{thm_fano_maximal}, let us
prepare for the smoothing-out argument that will appear therein.
For any $x\in\mathcal{X}$, denote by $(T_{x,t})_{t\ge 0}$ the simple
Markov semigroup on $\mathcal{Y}$ with stationary measure $P_{Y|X=x}$:
\begin{align}
	T_{x,t}f:=e^{-t}f+(1-e^{-t})P_{Y|X=x}(f).
\end{align}
We denote by
\begin{equation}
	T_{x^n,t} := T_{x_1,t}\otimes\cdots\otimes T_{x_n,t}
\label{e_prodsem}
\end{equation}
the corresponding product semigroup with stationary measure
$P_{Y^n|X^n=x^n}$.

Unlike the simpler setting of Section \ref{sec_pre},  the semigroup $T_{x^n,t}$ depends on $x^n$.
To
work around this issue, we introduce a linear operator $\Lambda_t:\mathcal{H}_+(\mathcal{Y}^n)\to
\mathcal{H}_+(\mathcal{Y}^n)$ that
dominates $T_{x^n,t}$ uniformly over $x^n$:
\begin{align}
	\Lambda_t :=
	[e^{-t}+\alpha(1-e^{-t})\nu]^{\otimes n}.
\label{e_semigroupM}
\end{align}
Note that $\Lambda_t$ is not a Markov semigroup; in particular, its
total mass satisfies
\begin{align}
	\Lambda_t1
	=\left(e^{-t}+\alpha(1-e^{-t})\right)^n
	\le e^{(\alpha-1)nt}.
\label{e8}
\end{align}
However, $\Lambda_t$ dominates the semigroups $T_{x^n,t}$ in the following
sense.

\begin{lem}\label{lem_semigroupmaximal}
Assume \eqref{e_assump} holds.
Then for any $f\in\mathcal{H}_+(\mathcal{Y}^n)$
\begin{align}
	\sup_{x^n\in\mathcal{X}^n}T_{x^n,t}f
	\le
	\Lambda_t f.
\label{e11}
\end{align}
\end{lem}

\begin{proof}
The result follows immediately from the definition of $T_{x^n,t}$,
$\Lambda_t$, and $\alpha$ upon noting that $P_{Y|X=x}(g) \le \alpha
\nu(g)$ whenever $g\ge 0$.
\end{proof}

We are now ready for the proof of Theorem \ref{thm_fano_maximal}.

\begin{proof}[Proof of Theorem \ref{thm_fano_maximal}]
Let $f_m:=1_{\mathcal{D}_m}$ and $t>0$ to be optimized later.
Note that
\begin{align}
	I(W;Y^n)
	&=
	\frac{1}{M}\sum_{m=1}^M D(P_{Y^n|X^n=c_m}\|P_{Y^n})
	\\
	&\ge
	\frac{1}{M}\sum_{m=1}^M
	P_{Y^n|X^n=c_m}(\ln \Lambda_t f_m)
	-\frac{1}{M}\sum_{m=1}^M\ln P_{Y^n}(\Lambda_t f_m)
\label{e_dv}
\end{align}
where \eqref{e_dv} is from the variational formula \eqref{e_var}.

We can lower bound the first term in \eqref{e_dv} as follows:
\begin{align}
\nonumber
	&\frac{1}{M}\sum_{m=1}^M
	\ln \|\Lambda_tf_m\|_{L^0(P_{Y^n|X^n=c_m})}
\\ 	&\ge
	\frac{1}{M}\sum_{m=1}^M\ln \|T_{c_m,t}f_m\|_{L^0(P_{Y^n|X^n=c_m})}
\label{e20}
\\ 	&\ge
	\frac{1}{1-e^{-t}}\,
	\frac{1}{M}\sum_{m=1}^M\ln P_{Y^n|X^n=c_m}(f_m)
\label{e21}
\\
	&\ge
	-\left(\frac{1}{t}+1\right)
	\ln \frac{1}{1-\epsilon}
\label{e37}
\end{align}
where \eqref{e20} is from \eqref{e11};
\eqref{e21} is from Theorem~\ref{thm_rhc} and $f\in[0,1]$;
and \eqref{e37} is from \eqref{e_geo} and $e^t\ge 1+t$.
For the second term in \eqref{e_dv}, we can estimate
\begin{align}
-\frac{1}{M}\sum_{m=1}^M
\ln P_{Y^n}(\Lambda_tf_m)
&\ge
-
\ln P_{Y^n}\left(\frac{1}{M}\sum_{m=1}^M \Lambda_tf_m\right)
\label{e54}
\\
&\ge \ln M-(\alpha-1)nt
\label{e54b}
\end{align}
using Jensen's inequality, $\sum_mf_m\le 1$, and \eqref{e8}.
Combining \eqref{e_dv}, \eqref{e37}, and \eqref{e54b} and optimizing over
$t>0$ concludes the proof.
\end{proof}

\subsection{Gaussian case}\label{sec_gfano}

The strong Fano inequality of the previous section applies in principle to
arbitrary alphabets $\mathcal{X},\mathcal{Y}$. However, in certain
situations the bounded density assumption \eqref{e_assump} may be overly
restrictive, particularly when the alphabets are unbounded. In this
section, we will adapt the argument of the previous section to the
prototypical example where this issue arises, namely, Gaussian channels.
Beside the intrinsic utility of the result, this will illustrate the
broader applicability of the smoothing-out method beyond the simple
semigroup setting exploited so far.

Throughout this section, we consider the random transformation
$P_{Y|X=x}=\mathcal{N}(x,1)$ from $\mathbb{R}$ to $\mathbb{R}$, so that
the corresponding memoryless channel is defined by
$P_{Y^n|X^n=x^n}=\mathcal{N}(x^n,\mathbf{I}_n)$. In this setting, we have
the following Gaussian analogue of Theorem~\ref{thm_fano_maximal}. For
later applications to broadcast channels, we consider a slight extension
of the setting of Theorem~\ref{thm_fano_maximal} to allow stochastic
encoders that use $\log_2 L$ random bits.

\begin{thm}\label{thm_gaussianfano}

Let $P_{Y|X=x}=\mathcal{N}(x,1)$.
Consider any encoding function $\phi:[M]\times[L]\to\mathbb{R}^n$
and disjoint decoding sets
$\mathcal{D}_1,\ldots,\mathcal{D}_M\subseteq\mathbb{R}^n$.
Suppose that
\begin{align}
	\prod_{w,v}P_{Y^n|X^n=\phi(w,v)}^{\frac{1}{ML}}[\mathcal{D}_w]
	\ge 1-\epsilon.
\label{e64}
\end{align}
Then
\begin{align}
	\ln M \le
	I(W;Y^n)
	+\sqrt{2\ln\frac{1}{1-\epsilon}}\sqrt{n}
	+\ln\frac{1}{1-\epsilon},
\end{align}
 where $Y^n$ is the output of the channel $P_{Y^n|X^n}$ with input
$X^n=\phi(W,V)$, and $(W,V)$ is equiprobable on $[M]\times[L]$. 
\end{thm}

In the Gaussian setting, it is natural to work not with simple
semigroups but rather with the \emph{Ornstein-Uhlenbeck semigroup}
with stationary measure $\mathcal{N}(x^n,\mb{I}_n)$:
\begin{align}
	T_{x^n,t}f(y^n)
	:=\mathbb{E}[f(e^{-t}y^n+(1-e^{-t})x^n
	+\sqrt{1-e^{-2t}}V^n)],
	\label{e_tx}
\end{align}
where $V^n\sim\mathcal{N}(0^n,\mb{I}_n)$.
In this setting, Borell \cite{borell1982positivity} showed that
reverse hypercontractivity
\eqref{e_RHC} holds under the even weaker assumption
$t\ge\frac{1}{2}\ln\frac{1-q}{1-p}$.

The proof will proceed a little differently than in Section \ref{sec_fano_d}.
Here, the analogue of $\Lambda_t$ in \eqref{e_semigroupM} is simply
$T_{0^n,t}$. Instead of Lemma~\ref{lem_semigroupmaximal},
we will exploit a simple change-of-variable formula: for any $f\ge 0$, $t>0$ and $x^n\in \mathbb{R}^n$,
we have
\begin{align}
	P_{Y^n|X^n=x^n}(\ln T_{0^n,t} f)
	=P_{Y^n|X^n=e^{-t}x^n}(\ln T_{e^{-t}x^n,t} f),
\label{e60}
\end{align}
which can be verified directly from  \eqref{e_tx}.

\begin{proof}[Proof of Theorem \ref{thm_gaussianfano}]

Let $f_w=1_{\mathcal{D}_w}$ for $w\in[M]$. Define also $\bar{X}^n=e^tX^n$,
and denote by $\bar{Y}^n$ the output of the channel $P_{Y^n|X^n}$ with
input $X^n\leftarrow\bar X^n$. Note that
\begin{align}
	& D(P_{\bar{Y}^n|W=w}\|P_{\bar{Y}^n})
\nonumber
\\
	&\ge
	\frac{1}{L}\sum_{v=1}^L
	P_{Y^n|X^n=e^t\phi(w,v)}
	(\ln T_{0^n,t}f_w)
	-\ln P_{\bar{Y}^n}(T_{0^n,t}f_w)
\label{e67}
\\
	&=\frac{1}{L}\sum_{v=1}^L
	P_{Y^n|X^n=\phi(w,v)}
	(\ln T_{\phi(w,v),t}f_w)
	-\ln P_{\bar{Y}^n}(T_{0^n,t}f_w)\label{e68}
\end{align}
for each $w$, where the key step \eqref{e68} used \eqref{e60}.

The summand in the first term of the right side of \eqref{e68} can be bounded using reverse
hypercontractivity for the Ornstein-Uhlenbeck semigroup \eqref{e_tx} as

\begin{align}
	P_{Y^n|X^n=\phi(w,v)}
	(\ln T_{\phi(w,v),t}f_w)
	&\ge \frac{1}{1-e^{-2t}}\ln
	P_{Y^n|X^n=\phi(w,v)}(f_w).
\end{align}
Therefore, using assumption \eqref{e64}, we obtain
\begin{align}
	\frac{1}{ML}\sum_{w,v} P_{Y^n|X^n=\phi(w,v)}
	(\ln T_{\phi(w,v),t}f_w)
	\ge -\frac{1}{1-e^{-2t}}\ln\frac{1}{1-\epsilon}.
\end{align}
On the other hand, by Jensen's inequality,
\begin{align}
	\frac{1}{M}\sum_{w=1}^M\ln P_{\bar{Y}^n}(T_{0^n,t}f_w)
	&\le
	\ln P_{\bar{Y}^n}\left(
	\frac{1}{M}\sum_{w=1}^M
	T_{0^n,t}f_w\right)
	\le \ln \frac{1}{M},
\end{align}
where we used $\sum_w f_w\le 1$. Thus, averaging over $w$
in \eqref{e68}, we find
\begin{align}
	\ln M \le
	I(W;\bar{Y}^n)
	+\frac{1}{1-e^{-2t}}\ln\frac{1}{1-\epsilon}.
\label{e_47}
\end{align}
We must now bound $I(W;\bar Y^n)$ in terms of $I(W;Y^n)$.
To this end,
let $G^n\sim \mathcal{N}(0^n,\mb{I}_n)$ be independent of $X^n$ and
note that
\begin{align}
	h(\bar{Y}^n)
	&=h(e^tX^n+G^n)
	\\
	&=h(X^n+e^{-t}G^n)+nt
	\\
	&\le h(Y^n)+nt,
\label{e_50}
\end{align}
where \eqref{e_50} can be seen from the entropy power inequality.
On the other hand,
\begin{align}
\nonumber
	& h(\bar{Y}^n|W=w)-h(Y^n|W=w)
	\\ &=I(\bar{Y}^n;X^n|W=w)
	-I(Y^n;X^n|W=w) \ge 0
\label{e_46}
\end{align}
for every $w$,
where the equality follows as $h(\bar{Y}^n|X^n,W=w)=h(Y^n|X^n,W=w)=h(G^n)$
and nonnegativity can be seen from \cite[Theorem~1]{guo2005mutual}.
Combining \eqref{e_47}, \eqref{e_50}, and \eqref{e_46} and using
$e^t\ge 1+t$ yields
\begin{align}
	\ln M \le
	I(W;Y^n)
	+\bigg(\frac{1}{2t}+1\bigg)
	\ln\frac{1}{1-\epsilon} + nt,
\end{align}
and the conclusion follows by optimizing over $t$.
\end{proof}

\begin{rem}
If the geometric average criterion \eqref{e64} is replaced by the
stronger maximal error criterion, a Gaussian Fano inequality with
$\sim\sqrt{n}$ second-order term can also be obtained using the
information spectrum method and the Gaussian Poincar\'e inequality; see
\cite[Theorem~8]{polyanskiy2014empirical}. However, beside the stronger
assumption, the resulting bound is much less explicit and does not recover
the correct dependence of the second-order term on the error probability
as $\epsilon\to 1$.
\end{rem}

\subsection{Application: output distribution of good channel codes}\label{sec_output}

As a very first application, let us note that Theorem 
\ref{thm_fano_maximal} effortlessly gives a sharp non-asymptotic converse 
for the basic data transmission example of Section \ref{sec:exist} for 
discrete memoryless channels. Since by the data processing inequality 
$$
	I(W;Y^n) \leq 
	n\sup_{P_X}I(X;Y) =: nC,
$$ 
the maximal code size satisfies 
\begin{equation}
\label{eq_simple_data}
	\ln M^* (n, \epsilon)  \le nC + 2\sqrt{|\mathcal{Y}|\ln\frac{1}{1-\epsilon}}\sqrt{n}
	+\ln\frac{1}{1-\epsilon},
\end{equation}
where we chose the reference measure $\nu$ in \eqref{e_assump} to be 
uniform on $\mathcal{Y}$. 

It was observed in \cite{han1993approximation,empirical} that if a code 
is close
to achieving the upper bound \eqref{eq_simple_data} , this places a strong
constraint on what the transmitted message looks like. In particular, it 
was shown there that if an $(n,M,\epsilon)$-code satisfies
\begin{equation}
	\ln M = nC + o(n),
\end{equation}
then it is necessarily the case that
\begin{equation}
	\frac{1}{n}D(P_{Y^n}\|P_{Y^n}^\star) = o(\epsilon),
\end{equation}
where 
$P_{Y^n}$ is the channel output distribution induced by the codebook 
and $P_{Y^n}^{\star}:=P_Y^{\star\otimes n}$ and $P_Y^{\star}$ is
the unique maximal mutual information output distribution (note that corresponding optimal input distribution need not be unique).
That is, the output distribution of a good code must approximate the 
capacity-achieving output distribution in the small error probability 
regime. Such a necessary condition for good channel codes sometimes 
provides guidelines for their design. 

Theorem \ref{thm_fano_maximal} enables us to develop a sharp quantitative 
form of this phenomenon for \emph{any} discrete memoryless channel.

\begin{thm}\label{thm_approx}
Consider a discrete memoryless channel $P_{Y|X}$ with capacity $C$. An 
$(n,M,\epsilon)$ code (under the maximal error criterion) satisfies
\begin{align}
	D(P_{Y^n}\|P_{Y^n}^{\star})
	\le
	nC-\ln M +
	2\sqrt{|\mathcal{Y}|\ln\frac{1}{1-\epsilon}}\sqrt{n}
	+\ln\frac{1}{1-\epsilon}
\label{e_approx}
\end{align}
where $P^{\star}_{Y^n}$ is the unique capacity achieving output 
distribution defined above.
\end{thm}

Two second-order bounds on the approximation error 
$D(P_{Y^n}\|P_{Y^n}^{\star})$ were previously derived in 
\cite{polyanskiy2014empirical}. In 
\cite[Theorem~7]{polyanskiy2014empirical}, an analogue of Theorem 
\ref{thm_approx} is obtained using the blowing-up method; as usual, this 
results in a suboptimal second-order term. However, in 
\cite[Theorem~6]{polyanskiy2014empirical} (see 
\cite[Theorem~3.6.6]{raginsky2014concentration} for a sharper formulation) 
a result analogous to Theorem \ref{thm_approx} is obtained with sharp
second-order term under the following assumption (the 
Burnashev condition \cite{burnashev1976data}):
\begin{align}
	\sup_{x,x'}\left\|\frac{dP_{Y|X=x}}{dP_{Y|X=x'}}
	\right\|_{\infty}<\infty.
\label{e_bur}
\end{align}
Theorem \ref{thm_approx} shows that the Burnashev condition is 
unnecessary: the result holds for any discrete memoryless channel. Its 
proof is completely straightforward; one simply follows the steps in 
\cite[(64)-(66)]{polyanskiy2014empirical} using the optimal second-order
Fano inequality (Theorem~\ref{thm_fano_maximal}) in lieu of the blowing-up
argument.

\begin{proof}[Proof of Theorem \ref{thm_approx}]
It suffices to note that
\begin{align}
	D(P_{Y^n}\|P_{Y^n}^{\star}) 
	&= D(P_{Y^n|X^n}\|P_{Y^n}^{\star}|P_{X^n}) \label{gold}
	-I(X^n;Y^n) \\ &\le nC - I(X^n;Y^n), \label{gold00}
\end{align}
where \eqref{gold} is a well-known identity, and \eqref{gold00} follows 
from 
the saddle point property of the capacity-achieving distribution 
\cite[eq.\ (19)]{polyanskiy2014empirical}.
It remains to bound $I(X^n;Y^n)$ using
Theorem~\ref{thm_fano_maximal}, where we choose $\nu$ in \eqref{e_assump} 
to be uniform on $\mathcal{Y}$.
\end{proof}

\begin{rem}
We have stated the above form of Theorem \ref{thm_approx} for simplicity. 
Without any additional difficulty, the finite alphabet assumption may be 
weakened to a bounded density assumption \eqref{e_assump} and the maximal 
error criterion may be weakened to the geometric average 
criterion \eqref{e_geo}. A Gaussian counterpart of 
Theorem~\ref{thm_approx} can also be obtained in a similar manner using 
Theorem~\ref{thm_gaussianfano}, strengthening the result of \cite[Theorem 
8]{polyanskiy2014empirical} to the geometric average criterion.
\end{rem}
\begin{rem}
By codebook expurgation (cf.\ Remark \ref{rem_expur}), 
one can slightly modify the strong channel coding 
converse \eqref{eq_simple_data} to remain valid even under the average
error criterion. In contrast, the proof of Theorem \ref{thm_approx} relied
on the sharper bound on $I(W;Y^n)$ given in Theorem 
\ref{thm_fano_maximal}, and thus the codebook expurgation method does not
apply in this setting. It is shown in \cite{allerton_lcv2017} that the 
geometric average criterion developed here is (in a particular sense) the 
weakest criterion under which Theorem \ref{thm_approx} may hold, so that 
our approach is optimal also in this sense.
\end{rem}

\subsection{Application: broadcast channels}\label{sec_bc}

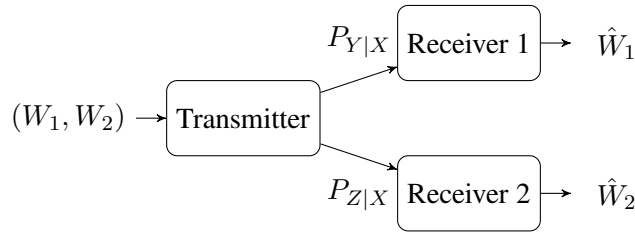
\begin{figure}[t]
  \centering
\begin{tikzpicture}
[node distance=0.6cm,minimum height=10mm,minimum width=12mm,arw/.style={->,>=stealth'}]
  \node[rectangle,draw,rounded corners] at (0,0) (T) {Transmitter};
  \node[rectangle,draw,rounded corners] at (3,1)(Y) {Receiver 1};
  \node[rectangle,draw,rounded corners] at (3,-1)(Z){Receiver 2};
  \node[rectangle] (W) [left =0.4cm of T] {$(W_1,W_2)$};
  \node[rectangle] (W1) [right =0.4cm of Y] {$\hat{W}_1$};
  \node[rectangle] (W2) [right =0.4cm of Z] {$\hat{W}_2$};

  \draw [arw] (W) to node[midway,above]{} (T);
  \draw [arw] (Y) to node[midway,right]{} (W1);
  \draw [arw] (Z) to node[midway,right]{} (W2);
  \draw [arw] (T) to node[midway,above]{$P_{Y|X}$} (Y);
  \draw [arw] (T) to node[midway,below]{$P_{Z|X}$} (Z);
\end{tikzpicture}
\label{f_agk}
\caption{The broadcast channel}
\end{figure}
Determining the capacity region of general broadcast channels is a 
long-standing open problem in information theory \cite{el2011network}. The 
general setting is illustrated in Figure~\ref{f_agk}. A transmitter wishes 
to send messages $(W_1,W_2)$, equiprobable on $[M_1]\times[M_2]$, to two 
respective receivers, using $n$ repetitions of the random transformations 
$P_{Y|X}$ and 
$P_{Z|X}$. A code consisting of codewords $c_{u,v}\in\mathcal{X}^n$ (with 
$u\in[M_1]$, $v\in[M_2]$) is said to be an 
\emph{$(n,M_1,M_2,\epsilon)$-code} under the average error criterion, if the receivers can reconstruct 
$\hat{W}_1$ and $\hat{W}_2$ from the outputs of their respective 
channels with error probabilities
\begin{align}
	\mathbb{P}[\hat{W}_1\neq W_1]\le \epsilon,\qquad
	\mathbb{P}[\hat{W}_2\neq W_2]\le \epsilon.
\label{eq_bavge}
\end{align}
The problem of understanding precisely what codebook sizes $M_1,M_2$ can 
be transmitted with a given error probability remains, in general, an open 
problem. However, we will consider the case of \emph{degraded broadcast 
channels} which is much better understood. The additional structural 
assumption in this case is that $P_{Z|X}$ equals the concatenation of 
$P_{Y|X}$ and a certain random transformation $P_{Z|Y}$ (or the other way 
around). In particular, in the setting of Gaussian broadcast channels, 
this additional assumption is always satisfied.

The aim of this section is to show how to obtain second-order converses for
both Gaussian and discrete degraded broadcast channels using our methods,
improving the best previously known bounds. To avoid digressions
unrelated to the topic of this paper, we will only sketch 
how Theorems \ref{thm_fano_maximal} and \ref{thm_gaussianfano} enter
the proofs. The subsequent manipulations of the first-order terms (using the chain rule and the entropy power 
inequality) are standard; we refer to \cite{el2011network} for the 
omitted steps and for matching achievability results. 
The maximal error criterion for the broadcast channel reads as
\begin{align}
	\max_{w_1,w_2}\mathbb{P}[\hat{W}_1\neq W_1|(W_1,W_2)=(w_1,w_2)]\le \epsilon,
	\\
	\max_{w_1,w_2}\mathbb{P}[\hat{W}_2\neq W_2|(W_1,W_2)=(w_1,w_2)]\le \epsilon.
	\end{align}
The reduction to the average error criterion for the broadcast channel can also be done by a
codebook expurgation argument, albeit more sophisticated than the single-user setting; see \cite[Problem~8.11]{el2011network}.
	
We now introduce the following geometric average decodability criterion for the broadcast channel,
which interpolates the average and the maximal error criteria above (see Remark \ref{rem_expur}).
\begin{align}
	\prod_{w_1,w_2}\mathbb{P}^{\frac{1}{M_1M_2}}[\hat{W}_1= W_1|(W_1,W_2)=(w_1,w_2)]\ge 1-\epsilon,
	\\
	\prod_{w_1,w_2}\mathbb{P}^{\frac{1}{M_1M_2}}[\hat{W}_2= W_2|(W_1,W_2)=(w_1,w_2)]\ge 1-\epsilon.
	\end{align}
The geometric average criterion integrates seamlessly with our proof, which does not involve an expurgation argument.

We begin by stating the Gaussian result. In the stationary 
memoryless Gaussian broadcast channel, it is assumed that 
$P_{Y^n|X^n}=P_{Y|X}^{\otimes n}$ and $P_{Z^n|X^n}=P_{Z|X}^{\otimes n}$
with 
$P_{Y|X=x}=\mathcal{N}(x,\sigma_1^2)$ and 
$P_{Z|X=x}=\mathcal{N}(x,\sigma_2^2)$. We assume moreover that the 
codewords $c\in\mathbb{R}^n$ must satisfy the power constraint 
$\|c\|^2\le nP$. The signal-to-noise ratios (SNR) of the two channels 
are then defined by $S_i:=P/\sigma_i^2$.

\begin{thm}\label{thm_gbc}
Consider a stationary memoryless Gaussian broadcast channel 
with SNRs $S_1,S_2\in (0,\infty)$.
Suppose there exists an $(n,M_1,M_2,\epsilon)$-code (under the geometric average criterion). Then
\begin{align}
	\ln M_1&\le n{\sf C}(\alpha S_1)
	+\sqrt{2n\ln\frac{1}{1-\epsilon}}
	+\ln\frac{1}{1-\epsilon},
\\
	\ln M_2&\le n{\sf C}\left(\frac{(1-\alpha)S_2}{\alpha S_2+1}\right)
	+\sqrt{2n\ln\frac{1}{1-\epsilon}}
	+\ln\frac{1}{1-\epsilon}
\label{e_gaussianBC}
\end{align}
for some $\alpha\in [0,1]$, where ${\sf C}(t):=\frac{1}{2}\ln(1+t)$.
\end{thm}

\begin{proof}
We first use Theorem~\ref{thm_gaussianfano} to obtain\footnote{Note that 
	as no power constraint is assumed in Theorem 
	\ref{thm_gaussianfano}, its conclusion extends verbatim to 
	channels with arbitrary positive variance by scaling.}
\begin{align}
	\ln M_1&\le I(W_1;Y^n|W_2)
	+\sqrt{2n\ln\frac{1}{1-\epsilon}}
	+\ln\frac{1}{1-\epsilon},
	\label{e344}
\\
	\ln M_2&\le I(W_2;Z^n)
	+\sqrt{2n\ln\frac{1}{1-\epsilon}}
	+\ln\frac{1}{1-\epsilon}.
\end{align}
To see \eqref{e344}, we first apply Theorem~\ref{thm_gaussianfano} to obtain:
\begin{align}
\ln M_1&\le I(W_1;Y^n|W_2=w_2)
	+\sqrt{2n\ln\frac{1}{1-\epsilon_{w_2}}}
	+\ln\frac{1}{1-\epsilon_{w_2}},
	\label{e349}
\end{align}
where we defined
$$\epsilon_{w_2}=	1-\prod_{w_1}\mathbb{P}^{\frac{1}{M_1}}[\hat{W}_1= W_1|(W_1,W_2)=(w_1,w_2)].$$
Then we obtain \eqref{e344} by averaging \eqref{e349} over $w_2$ and applying Jensen's inequality:	
\begin{align}
\frac{1}{M_2}\sum_{w_2}\sqrt{2n\ln\frac{1}{1-\epsilon_{w_2}}}
&\le
\sqrt{2n\frac{1}{M_2}\cdot\sum_{w_2}\ln\frac{1}{1-\epsilon_{w_2}}}
\\
&=\sqrt{2n\ln\frac{1}{\prod_{w_2}(1-\epsilon_{w_2})^{\frac{1}{M_2}}}}
\\
&\le \sqrt{2n\ln\frac{1}{1-\epsilon}},
\end{align}
where we used the fact that
$$\prod_{w_2}(1-\epsilon_{w_2})^{\frac{1}{M_2}}
=
\prod_{w_1,w_2}\mathbb{P}^{\frac{1}{M_1M_2}}[\hat{W}_1= W_1|(W_1,W_2)=(w_1,w_2)]
\ge 1-\epsilon.
$$
Now it remains to bound the mutual 
informations $I(W_1;Y^n|W_2)$ and $I(W_2;Z^n)$ by the respective 
capacities; the proof of this part of the argument is identical that of 
the weak converse \cite[Theorem~5.3]{el2011network}, and we omit the 
details.
\end{proof}

The analogous result in the discrete case is as follows.

\begin{thm}\label{thm_discreteBC}
Consider a degraded discrete memoryless broadcast channel
$(P_{Y|X},P_{Z|X})$ and $(n,M_1,M_2,\epsilon)$-code
(under the geometric average criterion). Then
\begin{align}
	\ln M_1
	&\le n I(X;Y|U)+2\sqrt{|\mathcal{Y}|n\ln\frac{1}{1-\epsilon}}
	+\ln\frac{1}{1-\epsilon},
\\
	\ln M_2
	&\le
	n I(U;Z)+2\sqrt{|\mathcal{Z}|n\ln\frac{1}{1-\epsilon}}
	+\ln\frac{1}{1-\epsilon}
\end{align}
for some distribution $P_{UX}$.
\end{thm}

\begin{proof}
The steps are identical to those of Theorem \ref{thm_gbc}, modulo the
(trivial) generalization of Theorem~\ref{thm_fano_maximal} to allow 
stochastic encoders as we did in Theorem~\ref{thm_gaussianfano}.
The omitted manipulations of the mutual information terms follow 
identically the proof of the weak converse in this setting
\cite[Theorem~5.2]{el2011network}.
\end{proof}

Theorem~\ref{thm_discreteBC} appears to be the first second-order converse for 
the discrete degraded broadcast channel, improving the suboptimal converse obtained by the 
blowing-up method \cite{ahlswede_bounds_cond1976} (cf.\ 
\cite[Theorem~3.6.4]{raginsky2014concentration}). Our approach further 
extends beyond finite alphabets to the bounded density setting 
\eqref{e_assump}. For the Gaussian broadcast channel, a strong converse 
with $\sim\sqrt{n}$ second-order term was previously obtained in 
\cite{fong2015proof} via the information spectrum method. The present 
bound \eqref{e_gaussianBC} is sharper; in particular, it has the correct 
asymptotic dependence not only on $n \to \infty$ but also on the error probability $\epsilon \to 1$. 

\section{Second-Order Image Size Characterization}\label{sec_change}

The aim of this section is to develop the smoothing-out methodology for 
coding problems in network information theory that are proved by the 
``image size characterization'' method. This powerful machinery was 
introduced in the original work of Ahlswede, G\'acs and K\"orner 
\cite{ahlswede_bounds_cond1976} for proving a strong converse for source 
coding with side information (cf.~Section~\ref{sec_source}), and was 
further developed systematically in the classic monograph of Csisz\'ar and 
K\"orner \cite[chapters 15--16]{csiszar2011information} to enable the 
analysis of a wide variety of source and channel networks.

Unfortunately, the general formulation of the image size method may appear 
rather technical at first sight, particularly if the reader is unfamiliar 
with the classical theory developed in \cite{csiszar2011information}. To 
make the theory more accessible, we will first provide some motivating 
discussion and background in Section \ref{sec_syn}. The main results of 
this section, which give second-order forms of the image-size 
characterization method in discrete and Gaussian cases, are given in 
Sections \ref{sec_change_d} and \ref{sec_change_g}, respectively. Finally, 
we will illustrate the application of our framework in two representative 
applications in Sections \ref{sec_change_hypo} and \ref{sec_source}.

It should be emphasized that in order to keep this paper to a reasonable 
length, we treat here only the more basic setting of the image size 
problem that appears in the original paper 
\cite{ahlswede_bounds_cond1976}. However, our methods extend in full 
generality to the theory developed in \cite{csiszar2011information}, which 
yields the strong converse property of all source-channel networks with 
known first-order rate region. In the terminology of 
\cite{csiszar2011information}, the problem considered here only contains a 
``forward channel'' while the general theory allows for the addition of a 
``reverse channel''; in fact, besides yielding second-order converses 
and extending to general alphabets, our methods have the further advantage 
of simplifying certain technical aspects in dealing with the reverse 
channel. A detailed development of our framework for the most general form 
of the image size characterization problem may be found in \cite[Chapter 
5]{liuthesis}.

\subsection{Synopsis}
\label{sec_syn}

\subsubsection{A motivating example}
\label{sec_change_ex}

Before we introduce the general setting for image size characterization, 
it is instructive to motivate its form through a simple application. To 
this end, let us begin by describing a variant of the binary hypothesis 
testing problem, due to Ahlswede and Csisz\'ar 
\cite{ahlswede1986hypothesis}, in which the image size problem arises 
naturally. After we have developed the general results, we will return to 
this example and provide a second-order converse using our methods 
in Section \ref{sec_change_hypo}.

Suppose 

we are given a joint probability measure $Q_{XY}$, and we obtain $n$ 
independent observations drawn from either $Q_{XY}$ or from 
$Q_X\otimes Q_Y$.
We would like 
to construct an optimal hypothesis test between these two alternatives:
\begin{align}
{\rm H}_0&\colon Q_{XY}^{\otimes n},
\nonumber\\
{\rm H}_1&\colon Q_X^{\otimes n}Q_Y^{\otimes n}.
\end{align}
When stated in this form, this is just a special case of the general 
binary hypothesis testing problem discussed in Section \ref{sec_pre}. In 
particular, for any test 
$\mathcal{B}\subseteq\mathcal{X}^n\times\mathcal{Y}^n$ with error 
probability $\pi_{1|0}:=Q_{XY}^{\otimes n}[\mathcal{B}^c]\le\epsilon$, we 
would like to understand how small an error probability 
$\pi_{0|1}:=Q_X^{\otimes n}Q_Y^{\otimes n}[\mathcal{B}]$ can be achieved 
(we will focus here on deterministic tests for simplicity).
In other words, the aim of the binary hypothesis testing problem is to 
understand the quantity
\begin{equation}
	\min_{Q_{X^nY^n}[\mathcal{B}]\ge 1-\epsilon} 
	Q_{X^n}Q_{Y^n}[\mathcal{B}].
\end{equation}
Both a second-order converse (lower bound) and the matching achievability 
argument (upper bound) follow 
immediately from the results in Section \ref{sec_pre}.

We now, however, add a new ingredient. Suppose only $Y^n$ is directly 
revealed to the statistician; the data $X^n$ is observed elsewhere, and 
must be communicated to the statistician through a (noiseless) 
communication channel with limited capacity. That is, an encoder must 
encode the data $X^n$ into a message $W\in[M]$ that is transmitted 
noiselessly to the statistician; the capacity of the channel limits the 
size of the codebook to $M\in\mathbb{N}$. This situation is illustrated in 
Figure \ref{fig_hypo}.
The communication constraint may be described by 
restricting the original binary hypothesis testing problem to a 
special class of tests of the form
\begin{equation}
	\mathcal{B} = \bigcup_{i=1}^M (\mathcal{A}_i\times\mathcal{B}_i)
	\subseteq \mathcal{X}^n\times\mathcal{Y}^n,
\label{eq_cc}
\end{equation}
where $\mathcal{A}_i=\{W=i\}$ is the set of strings $X^n$ for 
which the message $i$ is transmitted to the statistician, and 
$\mathcal{B}_i$ is the set of strings $Y^n$ for which the statistician 
declares $\mathrm{H}_0$ given that the message $i$ was received.
In this setting, we aim to understand the interplay between the channel 
capacity and the attainable error probabilities of the test.
That is, the aim of the hypothesis testing problem with 
communication constraint is to understand the quantity
\begin{equation}
	\min_{Q_{X^nY^n}[\mathcal{B}]\ge 
	1-\epsilon,~\mathcal{B}\mathrm{~as~in~}\eqref{eq_cc}} 
	Q_{X^n}Q_{Y^n}[\mathcal{B}].
\label{eq_hypocc}
\end{equation}
In accordance with the theme of this paper, we will be concerned 
here with the converse direction only; the achievability argument may be 
found in \cite{ahlswede1986hypothesis}.
\begin{figure}[t]
  \centering
\begin{tikzpicture}
[node distance=0.6cm,minimum height=10mm,minimum width=12mm,arw/.style={->,>=stealth'}]
  \node[draw,rectangle] (D) {Statistician};
  \node[draw,rectangle] (E1) [left=1.4cm of D] {Encoder};
  \node[rectangle] (X) [left =0.5cm of E1] {};
  \node[rectangle] (Yh) [right =0.5cm of D] {$\mathrm{H}_0/\mathrm{H}_1$};
  \node (Y) [below =0.5cm of D,minimum height=0] {};

  \draw[->]  (Y) to node[midway,below]{$Y^n$} (D);
  \draw[->]  (E1) to node[midway,above]{$W$} (D);
  \draw[->]  (D) to node[midway,left]{} (Yh);
  \draw[->]  (X) to node[midway,left]{$X^n$} (E1);

\end{tikzpicture}
\caption{Hypothesis testing with a communication constraint.}
\label{fig_hypo}
\end{figure}
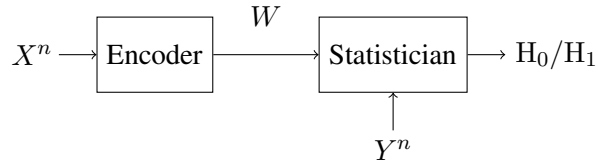

In Section \ref{sec_change_hypo}, we will give a second-order 
converse for the problem in \eqref{eq_hypocc}. For the purposes of the 
present motivating discussion, it will be convenient for simplicity to 
consider the stronger maximum error criterion variant\footnote{%
	In Section \ref{sec_change_hypo} we reduce the original 
	problem to such a variant by expurgation (cf.\ 
	Remark \ref{rem_expur}); however, for reasons explained in
	Section \ref{sec_typical}, this must be done carefully to attain 
	the correct rate.} 
of this problem:
how small can $Q_{X^n}Q_{Y^n}[\mathcal{B}|X^n]$ be made for 
communication-constrained tests \eqref{eq_cc}, given that
$Q_{X^nY^n}[\mathcal{B}|X^n]\ge 1-\epsilon$?
By applying \eqref{eq_cc}, this question can be reformulated as follows: if
we have
\begin{equation}
	\max_{x^n\in\mathcal{A}_i}
	Q_{Y^n|X^n=x^n}[\mathcal{B}_i^c]\le \epsilon,
\end{equation}
how small can $Q_{Y^n}[\mathcal{B}_i]$ be made? In this formulation, it is 
not entirely clear how the codebook size $M$ enters the picture. However, 
note that as there are only $M$ codewords, at least one of the codewords 
must have large probability $Q_{X^n}[\mathcal{A}_i]\ge 1/M$. Therefore, a 
lower bound for the maximal error variant of \eqref{eq_hypocc} naturally 
reduces to the problem of understanding the following quantity:
\begin{equation}
	\min_{Q_{X^n}[Q_{Y^n|X^n}[\mathcal{B}']\ge 1-\epsilon]\ge 1/M}
	Q_{Y^n}[\mathcal{B}'].
\label{eq_imagesize}
\end{equation}
The latter problem, which was originally considered in 
\cite{ahlswede_bounds_cond1976}, is the most basic form of the \emph{image 
size characterization problem} that is studied in this section.\footnote{
	The name arises from the fact that we think of $\mathcal{B}_i$
	as an ``$(1-\epsilon)$-image'' of $\mathcal{A}_i$ under the
	random transformation $Q_{Y^n|X^n}$ \cite{csiszar2011information}.
}

The motivation we have given here for the image size problem may appear 
rather specific to binary hypothesis testing with a communication 
constraint. However, just as ideas from binary hypothesis testing form the 
basis for a wide variety of coding problems, the image size problem turns 
out to arise in a broad range of problems in network information theory. 
After developing the general theory, we will discuss in some detail two 
applications: the present hypothesis testing problem with communication 
constraint (Section \ref{sec_change_hypo}) and source coding with 
compressed side information (Section \ref{sec_source}). Further 
applications are discussed in \cite[Chapter 16]{csiszar2011information}.
Another interesting application to the problem of common 
randomness generation is developed in detail in \cite[Section 4.4]{liuthesis}.

\subsubsection{The image size problem}
\label{sec_syn_img}

We now turn to the formulation of the basic objects and results that 
will be investigated in this section.

In the following, let $Q_{Y|X}$ be a given random transformation, and let 
$\nu$ and $\mu_n$ be positive measures on $\mathcal{Y}$ and 
$\mathcal{X}^n$, respectively. The general problem that will be 
investigated in this section is how to lower bound the $\nu^{\otimes 
n}$-measure of a set $\mathcal{A}\subseteq{Y}^n$ in terms of its
``$(1-\epsilon)$-preimage'' under $Q_{Y^n|X^n}:=Q_{Y|X}^{\otimes n}$.
Instead of the formulation in \eqref{eq_imagesize}, we find it 
convenient to introduce a Lagrange multiplier $c$ and investigate the
following unconstrained version of the problem: given any
$c>0$ and $\epsilon\in(0,1)$, find an upper bound on the quantity
\begin{equation}
	\sup_{\mathcal{A}\subseteq\mathcal{Y}^n}
	\{
	\ln\mu_n[Q_{Y^n|X^n}[\mathcal{A}]>1-\epsilon]
	-
	c\ln\nu^{\otimes n}[\mathcal{A}]	
	\}.
\label{eq_isc}
\end{equation}
To understand the behavior of this quantity, we begin by defining the 
information measure that controls it to first order.

\begin{defn}\label{defn_d}
Given positive measures $\mu$ on $\mathcal{X}$ and $\nu$ on $\mathcal{Y}$, 
a random transformation $Q_{Y|X}$, and constant $c>0$, we define
\begin{align}
	{\rm d}(\mu, Q_{Y|X}, \nu, c):=\sup_{P_X}
	\left\{cD(P_Y\|\nu)-D(P_X\|\mu)\right\}\label{e80}
\end{align}
where the probability measures $P_X,P_Y$ satisfy $P_X\to Q_{Y|X}\to 
P_Y$.
\end{defn}

\begin{rem}
The quantity \eqref{e80} was called 
``Brascamp-Lieb divergence'' in \cite{ISIT_lccv_smooth2016}.
Let us also note that the largest $c>0$ for which 
$d(Q_X,Q_{Y|X},Q_Y,c)=0$ is the reciprocal of the \emph{strong data 
processing constant} of $Q_{Y|X}$, cf.\ 
\cite{lccv2015}.
\end{rem}

To establish the natural connection between \eqref{eq_isc} and 
\eqref{e80}, let us begin by proving a weak converse. Evidently the method 
of proof has much in common with that of Lemma \ref{lem_weakbht} and is 
completely elementary (cf.\ \cite[(19)-(21)]{ahlswede_bounds_cond1976}).

\begin{lem}[\bf weak converse]\label{lem_weakimg}
If $\nu$ is a probability measure, then 
\begin{align}
	&\ln \mu_n[Q_{Y^n|X^n}[\mathcal{A}]>1-\epsilon]
	-
	c(1-\epsilon)\ln\nu^{\otimes n}[\mathcal{A}]
	\nonumber\\
	&\qquad\le {\rm d}(\mu_n, Q_{Y|X}^{\otimes n}, \nu^{\otimes n}, c)
	+c\ln 2
\label{e_change1}
\end{align}
for every set $\mathcal{A}\subseteq\mathcal{Y}^n$ and
$c>0$, $\epsilon\in(0,1)$, $Q_{Y|X}$, and positive measure $\mu_n$.
\end{lem}

\begin{proof}
Define the event $\mathcal{B}$ and probability measure $Q_{X^n}$ by
\begin{equation}
	\mathcal{B}=\{x^n\colon Q_{Y^n|X^n=x^n}[\mathcal{A}]>1-\epsilon\},
	\qquad
	Q_{X^n}[\mathcal{C}]=
	\frac{\mu_n[\mathcal{B}\cap\mathcal{C}]}{\mu_n[\mathcal{B}]},
\end{equation}
and let $Q_{Y^n}$ be defined by $Q_{X^n}\to Q_{Y^n|X^n}\to Q_{Y^n}$.
Then
\begin{equation}
	D(Q_{X^n}\|\mu_n)=\ln\frac{1}{\mu_n[\mathcal{B}]},
\end{equation}
while by the data processing inequality
\begin{align}
	D(Q_{Y^n}\|\nu^{\otimes n})
	&\ge Q_{Y^n}[\mathcal{A}]\ln\frac{Q_{Y^n}[\mathcal{A}]}
	{\nu^{\otimes n}[\mathcal{A}]}
	+
	Q_{Y^n}[\mathcal{A}^c]\ln\frac{Q_{Y^n}[\mathcal{A}^c]}
	{\nu^{\otimes n}[\mathcal{A}^c]}
\\
	&\ge
	(1-\epsilon)
	\ln\frac{1}{\nu^{\otimes n}[\mathcal{A}]}
	-h(Q_{Y^n}[\mathcal{A}^c]).
\end{align}
Thus,
\begin{align}
	{\rm d}(\mu_n,Q_{Y|X}^{\otimes n},\nu^{\otimes n})
	&\ge
	cD(Q_{Y^n}\|\nu^{\otimes n})
	-D(Q_{X^n}\|\mu_n)
\\
	&\ge
	c(1-\epsilon)\ln\frac{1}{\nu^{\otimes n}[\mathcal{A}]}-\ln\frac{1}{\mu_n[\mathcal{B}]}
	-c\ln 2,
\end{align}
which is \eqref{e_change1}.
\end{proof}

Lemma \ref{lem_weakimg} is a weak converse due to the extraneous factor 
$1-\epsilon$ in front of $\ln\nu^{\otimes n}[\mathcal{A}]$ in 
\eqref{e_change1}, which prevents us from attaining the optimal rate in 
the regime of nonvanishing error probability. Instead, we aim to 
prove a strong converse 
\begin{align}
	&\ln \mu_n[Q_{Y^n|X^n}[\mathcal{A}]>1-\epsilon]
	-
	c\ln\nu^{\otimes n}[\mathcal{A}]
	\nonumber\\
	&\qquad\le {\rm d}(\mu_n, Q_{Y|X}^{\otimes n}, \nu^{\otimes n}, c)
	+ o_\epsilon(n).
\label{eq_strongimg}
\end{align}
Ahlswede, G\'acs and K\"orner \cite{ahlswede_bounds_cond1976} used
the blowing-up method to establish \eqref{eq_strongimg} for 
finite alphabets. As usual, such an argument can only attain a 
suboptimal second-order term. In the follows subsections, we will use the 
smoothing-out method to obtain new sharp forms of
\eqref{eq_strongimg} that attain both the $O(\sqrt{n})$ 
second-order term, and that may be extended beyond the finite-alphabet 
setting. In particular, we develop the relevant theory for discrete 
channels in Section~\ref{sec_change_d}, and for Gaussian channels in 
Section~\ref{sec_change_g}.

\subsubsection{The first-order term}
\label{sec_typical}

Before we proceed to the main results of this section, we must discuss an 
independent issue that plays an important role in applications of image 
size characterizations. Consider the motivating example 
of hypothesis testing with a communication constraint discussed above.
A lower bound in \eqref{eq_imagesize} is readily obtained by applying
Lemma \ref{lem_weakimg} or \eqref{eq_strongimg} with $\mu_n=Q_X^{\otimes 
n}$. When all measures in ${\rm d}(\cdot)$ are product measures, it is 
readily evaluated by means of the \emph{tensorization property}
\begin{align}
	{\rm d}(Q_X^{\otimes n}, Q_{Y|X}^{\otimes n}, Q_Y^{\otimes n}, c)
	=n\,{\rm d}(Q_X, Q_{Y|X}, Q_Y, c)
\label{e_tens1}
\end{align}
(a simple proof may be found, for example, in \cite{lccv2015}). Thus,
to first order, the quantity defined in \eqref{eq_imagesize} grows 
linearly in $n$ with rate ${\rm d}(Q_X, Q_{Y|X}, Q_Y, c)$ (with the 
optimal choice of Lagrange multiplier $c$). This is in fact the correct 
growth rate of \eqref{eq_imagesize}, as is demonstrated in
\cite[Theorem 1]{ahlswede_bounds_cond1976}. Unfortunately, 
however, this quantity does \emph{not} give the correct first-order rate for 
the hypothesis testing problem with communication constraint: the best 
achievability result for the latter has a strictly smaller rate 
than is suggested by the above arguments.

It turns out that the origin of this inefficiency does not lie in 
\eqref{eq_strongimg} itself, but rather in the initial argument that
reduced the information-theoretic problem to the study of the quantity 
\eqref{eq_imagesize}. When the argument is developed 
in detail, we will see that it is not necessary to consider the probability 
of codewords under the full measure $Q_{X^n}$, but only under the 
restriction
of this measure to any set $\mathcal{C}$ of sufficiently high probability.
While the information-theoretic problem is not affected by throwing out
a small probability event, doing so has a significant effect on the 
resulting bounds. In particular,
as was noted in \cite{ahlswede_bounds_cond1976} in the finite alphabet 
setting, the best possible improvement to the first-order rate is obtained 
by choosing $\mu_n=Q_{X^n}|_{\mathcal{C}_n}$ to be the restriction to a set 
$\mathcal{C}_n$ of \emph{typical sequences}. As we will show in 
Appendices~\ref{app_sbl} and \ref{app_dgaussian} for the discrete and 
Gaussian cases, respectively,\footnote{The claim in \eqref{e_49} is the main theme of the conference paper 
\cite{ISIT_lccv_smooth2016} and the journal paper \cite{lccv2019smooth}. The proof of \eqref{e_49} in the present paper is based on the concentration of the empirical distribution of $X^n$, which is similar in spirit to \cite{ISIT_lccv_smooth2016}. The proof in \cite{lccv2019smooth} adopts a different convex analysis argument which establishes \eqref{e_49} for general distributions and also determines the exact prefactor in the $O(\sqrt{n})$ term. 
However, for conceptual simplicity we do not adopt the approach of \cite{lccv2019smooth} in the present paper.} such a choice of $\mu_n$ in fact gives rise to 
the estimate
\begin{align}
	{\rm d}(\mu_n, Q_{Y|X}^{\otimes n}, \nu^{\otimes n}, c)
	\le
	n\,{\rm d}^{\star}(Q_X, Q_{Y|X}, \nu, c)+O(\sqrt{n})
\label{e_49}
\end{align}
where the modified rate ${\rm d}^{\star}(\cdot)$ is defined as follows.

\begin{defn}\label{defn_dstar}
Given a probability measure $Q_X$ on $\mathcal{X}$, a positive measure
$\nu$ on $\mathcal{Y}$, a random transformation $Q_{Y|X}$, and 
constant $c>0$, we define
\begin{align}
\nonumber
	&{\rm d}^{\star}(Q_X, Q_{Y|X}, \nu, c) := \\
	&\qquad \sup_{P_{UX}\colon P_X=Q_X}
	\left\{cD(P_{Y|U}\|\nu|P_U)-D(P_{X|U}\|Q_X|P_U)\right\}
\label{eq_dstar}
\end{align}
where the joint distribution of $U,X,Y$ is given by 
$P_{UXY}:=P_{UX}Q_{Y|X}$.
\end{defn}

We observe that it follows immediately from Definitions \ref{defn_d} and 
\ref{defn_dstar} that ${\rm d}^{\star}(Q_X, Q_{Y|X}, \nu, c)\le {\rm d}(Q_X, 
Q_{Y|X}, \nu, c)$. Therefore, restricting to a typical set always results in 
an improved bound. It will turn out that the rate \eqref{eq_dstar} captures 
precisely the correct first-order behavior of the information-theoretic 
applications in which image size characterizations will be applied. 
Moreover, as we are able to establish in \eqref{e_49} a second-order term of 
order $O(\sqrt{n})$, the combination of such an estimate with the 
smoothing-out method enables us to attain second-order converses for 
applications of the image size problem. As the methods needed to establish 
\eqref{e_49} are unrelated to the main topic of this paper, we focus 
presently on the smoothing-out method and relegate the proof of
\eqref{e_49} to Appendix \ref{app_proofs}.

\subsection{Discrete case}\label{sec_change_d}

The aim of this section is to established a non-asymptotic form of 
the strong converse for image size characterization \eqref{eq_strongimg} 
in the finite alphabet setting. We will, in fact, initially work in the 
more general bounded density setting as in Section \ref{sec_fano_d}. We 
specialize subsequently to the finite alphabet setting only in order to 
achieve the appropriate characterization \eqref{e_49} of the first-order 
term.

For the time being, let $\mathcal{X},\mathcal{Y}$ be general alphabets, 
let $Q_{Y|X}$ be a random transformation from $\mathcal{X}$ to 
$\mathcal{Y}$, and let $\nu$ be a probability measure on $\mathcal{Y}$.
We define, as usual, $Q_{Y^n|X^n}:=Q_{Y|X}^{\otimes n}$, and introduce
the bounded density assumption
\begin{equation}
	\alpha :=
	\sup_{x\in\mathcal{X}}\left\|\frac{dQ_{Y|X=x}}{d\nu}\right\|_{\infty}
	\in [1,\infty).
\label{e_change_d_aspt}
\end{equation}
The basic result of this section is the following quantitative form of 
\eqref{eq_strongimg}.

\begin{thm}
\label{thm_change_bdd}
For any positive measure $\mu_n$ on $\mathcal{X}^n$,
$f\in \mathcal{H}_{[0,1]}(\mathcal{Y}^n)$, $\eta\in(0,1)$, and $c>0$,
we have
\begin{align}
	&\ln\mu_n[Q_{Y^n|X^n}(f)\ge \eta]
	-
	c\ln\nu^{\otimes n}(f)
	\nonumber\\
	&\qquad\le
	{\rm d}(\mu_n, Q_{Y^n|X^n}, \nu^{\otimes n}, c)+
	2c\sqrt{\ln\frac{1}{\eta}}\sqrt{n(\alpha-1)}
	+c\ln\frac{1}{\eta}.
\label{e14}
\end{align}
\end{thm}

The proof of Theorem~\ref{thm_change_bdd} relies on similar ideas to the 
proof of Theorem~\ref{thm_fano_maximal}. As a first step, we must develop 
a variational characterization for the quantity ${\rm d}(\cdot)$, 
generalizing the variational formula \eqref{e_var} for relative entropy.

\begin{prop}\label{prop_equiv}
In the setting of Definition~\ref{defn_d},
\begin{align}
	{\rm d}(\mu, Q_{Y|X}, \nu, c)
	=\sup_{f\in \mathcal{H}_+(\mathcal{Y})}
	\left\{\ln \mu(e^{cQ_{Y|X}(\ln f)})-c\ln \nu(f)\right\}.
\label{e53}
\end{align}
\end{prop}

\begin{proof}
There are several proofs in the literature (cf.\ \cite{lccv2015}); 
here we give a short proof using \eqref{e_var}. First,
for any $f$, define a probability measure $P_X$ by
\begin{equation}
	\frac{dP_X}{d\mu}=
	\frac{e^{cQ_{Y|X}(\ln f)}}{\mu(e^{cQ_{Y|X}(\ln f)})},
\end{equation}
and define $P_Y$ by $P_X\to Q_{Y|X}\to P_Y$. Then by \eqref{e_var},
we have
\begin{align}
	&\ln\mu(e^{cQ_{Y|X}(\ln f)})-c\ln \nu(f)
	\nonumber\\
	&\le \ln\mu(e^{cQ_{Y|X}(\ln f)})+cD(P_Y\|\nu)-cP_Y(\ln f)
	\\
	&=cD(P_Y\|\nu)-D(P_X\|\mu).
\end{align}
This proves the $\ge$ part of \eqref{e53}.

Conversely, for any $P_X$, we have using \eqref{e_var}
\begin{align}
	&cD(P_Y\|\nu)-D(P_X\|\mu)
	\nonumber\\
	&\le cD(P_Y\|\nu)-P_X(cQ_{Y|X}(\ln f))+\ln \mu(e^{cQ_{Y|X}(\ln f)})
	\\
	&=\ln\mu(e^{cQ_{Y|X}(\ln f)})-c\ln \nu(f)
\end{align}
where we chose $f=\frac{dP_Y}{d\nu}$.
This proves the $\le$ part of \eqref{e53}.
\end{proof}

We can now complete the proof of Theorem \ref{thm_change_bdd}. In the 
following, we will use the same semigroup device as in Section 
\ref{sec_fano_d}: we define the simple semigroup
\begin{align}
        T_{x,t}f:=e^{-t}f+(1-e^{-t})Q_{Y|X=x}(f)
\end{align}
and its product $T_{x^n,t} := T_{x_1,t}\otimes\cdots\otimes T_{x_n,t}$,
and we define the corresponding dominating operator $\Lambda_t$ 
according to \eqref{e_semigroupM}. We recall, in particular, the key 
properties \eqref{e8} and \eqref{e11} that will also be used in the 
present proof.

\begin{proof}[Proof of Theorem \ref{thm_change_bdd}]
We begin by noting that, by the variational principle given in
Proposition~\ref{prop_equiv}, we can estimate
\begin{align}
	\int\|g\|_{L^0(Q_{Y^n|X^n=x^n})}
	d\mu_n(x^n)
	&=\int e^{Q_{Y^n|X^n}(\ln g)} d\mu_n
	\\
	&\le
	e^{{\rm d}(\mu_n,Q_{Y^n|X^n},\nu^{\otimes n},c)}
	\|g\|_{L^{1/c}(\nu^{\otimes n})}
\label{e_daul}
\end{align}
for any function $g\in\mathcal{H}_+(\mathcal{Y}^n)$ for which the 
integrals are finite.

Now take $g=(\Lambda_tf)^c$ and observe that by \eqref{e8},
\begin{align}
	\|g\|_{L^{1/c}(\nu^{\otimes n})}
	&\le 
	e^{c(\alpha-1)nt}[\nu^{\otimes n}(f)]^c.
\label{e26a}
\end{align}
On the other hand,
\begin{align}
	&\int\|g\|_{L^0(Q_{Y^n|X^n=x^n})}
	d\mu_n(x^n)
	\nonumber\\
	&=
	\int\|\Lambda_tf\|_{L^0(Q_{Y^n|X^n=x^n})}^c
	d\mu_n(x^n)
\label{e420}
\\
	&\ge
	\int\|f\|_{L^{1-e^{-t}}(Q_{Y^n|X^n=x^n})}^c
	d\mu_n(x^n)\label{e24}
\\
	&\ge
	\int Q_{Y^n|X^n}(f)^{\frac{c}{1-e^{-t}}}
	d\mu_n\label{e_25}
\\
	&\ge 
	\eta^{\frac{c}{1-e^{-t}}}
	\mu_n[Q_{Y^n|X^n}(f)\ge \eta],
\label{e26}
\end{align}
where \eqref{e24} used \eqref{e11} and reverse hypercontractivity 
(Theorem~\ref{thm_rhc}); \eqref{e_25} used that $f\in[0,1]$; and
\eqref{e26} follows from Markov's inequality. Putting together the 
estimates \eqref{e_daul}, \eqref{e26a}, and \eqref{e26}, we have shown 
that for every $t>0$
\begin{align}
\nonumber
	& \ln\mu_n[Q_{Y^n|X^n}(f)\ge \eta]
	-
	c\ln \nu^{\otimes n}(f) \\
	&\qquad
	\le
	{\rm d}(\mu_n,Q_{Y^n|X^n},\nu^{\otimes n},c)
	+
	\frac{c}{1-e^{-t}}\ln\frac{1}{\eta}
	+
	c(\alpha-1)nt.
\label{e26opt}
\end{align}
The desired inequality \eqref{e14} follows using $\frac{1}{1-e^{-t}}\le
\frac{1}{t}+1$ and optimizing over $t$.
\end{proof}

As was explained in Section \ref{sec_syn_img}, it is essential for 
applications of Theorem~\ref{thm_change_bdd} to make a judicious choice of 
measure $\mu_n$ in order to attain the correct first-order rate of the 
information-theoretic problems of interest. To this end, we now specialize 
to the case of finite alphabets and state a ready-to-use result along 
these lines. The precise construction of $\mu_n$ may be found in Appendix 
\ref{app_sbl}.

\begin{cor}
\label{cor_change_discrete}
Let $|\mathcal{X}|<\infty$, $Q_X$ be a probability measure on 
$\mathcal{X}$, $\nu$ be a probability measure on $\mathcal{Y}$, and
$Q_{Y|X}$ be a random transformation. Define
\begin{align}
	\beta_X :=\frac{1}{\min_x Q_X(x)}\in [1,\infty),\label{e52}
\end{align}
and define $\alpha$ as in \eqref{e_change_d_aspt}. Then for any
$\delta\in (0,1)$\ and $n>3\beta_X\ln\frac{|\mathcal{X}|}{\delta}$,
we may choose a set $\mathcal{C}_n\subseteq\mathcal{X}^n$ with
$Q_X^{\otimes n}[\mathcal{C}_n]\ge 1-\delta$ such that
\begin{align}
	&\ln \mu_n[Q_{Y^n|X^n}(f)\ge \eta]
	-
	c\ln\nu^{\otimes n}(f)
\nonumber
\\
	&\qquad\le
	n\,{\rm d}^{\star}(Q_X,Q_{Y|X},\nu,c)
	+
	A\sqrt{n}
	+c\ln\frac{1}{\eta}
\label{e15}
\end{align}
for all $f\in \mathcal{H}_{[0,1]}(\mathcal{Y}^n)$, $c>0$, 
$\eta\in(0,1)$, where we defined
$\mu_n:=Q_X^{\otimes n}|_{\mathcal{C}_n}$ and
\begin{align}
\label{eq_cor_chg_disc_A}
	A&:=\ln (\alpha^c\beta_X^{c+1})
	\sqrt{3\beta_X\ln\frac{|\mathcal{X}|}{\delta}}
	+
	2c\sqrt{(\alpha-1)\ln\frac{1}{\eta}}.
\end{align}
\end{cor}

\begin{proof}
This follows immediately by combining Theorem \ref{thm_change_bdd}
and Theorem \ref{thm_sbl} in Appendix \ref{app_sbl} (note that 
$\alpha_Y$ defined in the Appendix satisfies $\alpha_Y\le\alpha$).
\end{proof}

While the formulation of Theorem \ref{thm_change_bdd} and Corollary 
\ref{cor_change_discrete} is closer in spirit to the formulation of the 
image-size characterization problem in 
\cite{ahlswede_bounds_cond1976,csiszar2011information}, it is worth noting 
that our approach naturally gives rise to a somewhat sharper variant of 
these results that can sometimes be used to obtain cleaner bounds in 
converse proofs. In particular, it is not really necessary to apply the 
Markov inequality \eqref{e26}; we may work directly in applications with 
the following smoother version of the problem (stated, for future 
reference, in the setting of Corollary \ref{cor_change_discrete}).

\begin{cor}\label{cor_alternative}
Let $\beta_X$, $\alpha$, $\delta$, $n$ and $\mu_n$ be as in  
Corollary~\ref{cor_change_discrete}. Then
\begin{align}
	&\ln\int Q_{Y^n|X^n}^{c\left(1+\tfrac{1}{t}\right)}(f)
	\,d\mu_n
	-c\ln\nu^{\otimes n}(f)
\nonumber\\
	&\le
	n\,{\rm d}^{\star}(Q_X,Q_{Y|X},\nu,c)
	+ c(\alpha-1)nt
	+\ln(\alpha^c\beta_X^{c+1})
	\sqrt{3n\beta_X\ln\frac{|\mathcal{X}|}{\delta}}
\label{e_425}
\end{align}
for every $c,t>0$ and $f\in\mathcal{H}_{[0,1]}(\mathcal{Y}^n)$.
\end{cor}

\begin{proof}
Simply omit the use of Markov's inequality \eqref{e26} and the subsequent 
optimization over $t$ in the proof of Theorem \ref{thm_change_bdd}.
\end{proof}

In the sequel, we will illustrate both 
Corollaries~\ref{cor_change_discrete} and \ref{cor_alternative} in 
applications.

\subsection{Gaussian case}\label{sec_change_g}

Beside giving rise to $O(\sqrt{n})$ second-order terms, another key 
advantage of the smoothing-out method is that it is applicable beyond the 
finite alphabet setting. We will presently further illustrate this feature 
by developing a Gaussian analogue of the image size theory of the previous 
section, opening the door to systematic extension of many applications of 
this methodology to the Gaussian setting. The basic result of this section 
is the following Gaussian version of Theorem \ref{thm_change_bdd}.

\begin{thm}
\label{thm_change_g}
Let $Q_{Y|X=x}=\mathcal{N}(x,1)$ and 
$\nu$ be Lebesgue on $\mathbb{R}$.
Then we have
\begin{align}
	&\ln\mu_n[Q_{Y^n|X^n}(f)>\eta]
	-c\ln\nu^{\otimes n}(f)
	\nonumber
	\\
	&\qquad\le {\rm d}(\mu_n,Q_{Y^n|X^n},\nu^{\otimes n},c)
	+c\sqrt{2n\ln\frac{1}{\eta}}
	+c\ln\frac{1}{\eta}
\label{e102}
\end{align}
for any positive measure $\mu_n$ on $\mathbb{R}^n$,
$f\in \mathcal{H}_{[0,1]}(\mathbb{R}^n)$, $\eta\in(0,1)$, and $c>0$.
\end{thm}

The proof uses some ideas that were introduced in Section~\ref{sec_gfano}. 
In particular, in this section $T_{x^n,t}$ will denote the 
Ornstein-Uhlenbeck semigroup \eqref{e_tx}, and we will again exploit 
heavily the change-of-variables formula \eqref{e60}.

\begin{proof}[Proof of Theorem \ref{thm_change_g}]
Denote by $\bar{\mu}_n$ the dilation of $\mu_n$ by the factor $e^t$, that 
is,
\begin{equation}
	\bar{\mu}_n[e^t\mathcal{C}]:=\mu_n[\mathcal{C}]
\label{e_98}
\end{equation}
for any set $\mathcal{C}$. Applying Proposition~\ref{prop_equiv}, we can 
estimate
\begin{align}
	\int e^{Q_{Y^n|X^n}(\ln g)}
	d\bar{\mu}_n
	\le
	e^{{\rm d}(\bar{\mu}_n,Q_{Y^n|X^n},\nu^{\otimes n},c)}
	\|g\|_{L^{1/c}(\mathbb{R}^n)}
\label{e62}
\end{align}
for any function $g\in\mathcal{H}_+(\mathbb{R}^n)$ for which the integrals 
are finite.

Now take $g=(T_{0^n,t}f)^c$ and observe that
\begin{align}
	\|g\|_{L^{1/c}(\mathbb{R}^n)}
	&= \|T_{0^n,t}f\|_{L^1(\mathbb{R}^n)}^{c}
	\\
	&= e^{cnt}\|f\|_{L^1(\mathbb{R}^n)}^{c}
	\label{e38}
	\\
	&=e^{cnt}[\nu^{\otimes n}(f)]^c,
	\label{e66}
\end{align}
where \eqref{e38} can be verified using Fubini's theorem and \eqref{e_tx}.
On the other hand,
\begin{align}
	\int e^{Q_{Y^n|X^n}(\ln g)}
	d\bar{\mu}_n
	&=
	\int e^{cQ_{Y^n|X^n=x^n}(\ln T_{0^n,t}f)}
	d\bar{\mu}_n(x^n)
\\
	&=\int
	e^{cQ_{Y^n|X^n=e^{-t}x^n}(\ln T_{{e^{-t}}x^n,t}f)}
	d\bar{\mu}_n(x^n)
\label{e_68}
\\
	&=\int
	\|T_{{e^{-t}}x^n,t}f\|_{L^0(Q_{Y^n|X^n=e^{-t}x^n})}^c
	d\bar{\mu}_n(x^n)
\\
	&\ge\int
	\|f\|_{L^{1-e^{-2t}}(Q_{Y^n|X^n=e^{-t}x^n})}^c
	d\bar{\mu}_n(x^n)
\label{e107}
\\
	&\ge\int
	\|f\|_{L^1(Q_{Y^n|X^n=e^{-t}x^n})}^{\frac{c}{1-e^{-2t}}}
	d\bar{\mu}_n(x^n)
\label{e108}
\\
&\ge
	\eta^{\frac{c}{1-e^{-2t}}}
	\bar{\mu}_n[x^n\colon Q_{Y^n|X^n=e^{-t}x^n}(f)>\eta]
\label{e42}
\\
	&= \eta^{
	\frac{c}{1-e^{-2t}}}
	\mu_n[Q_{Y^n|X^n}(f)>\eta],
\label{e9}
\end{align}
where \eqref{e_68} follows from the change of variables formula 
\eqref{e60};
\eqref{e107} is from \eqref{e_RHC} (with the sharp constant for Gaussian 
hypercontractivity, see secion \ref{sec_gfano});
\eqref{e108} used that $f\in[0,1]$; and
\eqref{e9} follows from the definition \eqref{e_98} of $\bar\mu_n$.

Combining \eqref{e62}, \eqref{e66}, and \eqref{e9}, we obtain
\begin{align}
	&\ln\mu_n[Q_{Y^n|X^n}(f)>\eta]
	-c\ln\nu^{\otimes n}(f)
\nonumber
\\
	&\qquad\le {\rm d}(\bar{\mu}_n,Q_{Y^n|X^n},\nu^{\otimes n},c)
	+\inf_{t>0}\left\{
	\frac{c}{1-e^{-2t}}
	\ln\frac{1}{\eta}
	+cnt\right\}
\\
	&\qquad\le {\rm d}(\bar{\mu}_n,Q_{Y^n|X^n},\nu^{\otimes n},c)
	+c\sqrt{2n\ln\frac{1}{\eta}}
	+c\ln\frac{1}{\eta},
\label{e46}
\end{align}
where \eqref{e46} used $\frac{1}{1-e^{-2t}}\le
\frac{1}{2t}+1$.

To conclude the proof, it remains only to show that
\begin{align}
	{\rm d}(\bar{\mu}_n,Q_{Y^n|X^n},\nu^{\otimes n},c)
	\le{\rm d}(\mu_n,Q_{Y^n|X^n},\nu^{\otimes n},c),
	\label{e48}
\end{align}
for which we use the following modification of the argument that led to
\eqref{e_46}. When $P_{X^n}$ and $\mu_n$ are both scaled by $e^t$, the
relative entropy $D(P_{X^n}\|\mu_n)$ remains unchanged. On the other hand,
as $D(P_{Y^n}\|\nu^{\otimes n})=-h(P_{Y^n})$, the same argument as was 
used in \eqref{e_46} yields $D(P_{\bar{Y}^n}\|\nu^{\otimes n})\le 
D(P_{Y^n}\|\nu^{\otimes n})$. Substituting these observations into 
Definition \ref{defn_d} concludes the proof of \eqref{e48}.
\end{proof}

We finally state a ready-to-use Gaussian analogue of Corollary 
\ref{cor_change_discrete}. The precise construction of $\mu_n$ may be 
found in Appendix \ref{app_dgaussian}.

\begin{cor}\label{cor7}
Let $Q_{XY}$ be any nondegenerate Gaussian measure, and
let $\nu$ be the Lebesgue measure on $\mathbb{R}$. Then for any 
$\delta\in(0,1)$ and $n\ge 20\ln\frac{2}{\delta}$, we may choose a set
$\mathcal{C}_n\subseteq\mathbb{R}^n$ with 
$Q_X^{\otimes n}[\mathcal{C}_n]\ge 1-\delta$ such that
\begin{align}
	&\ln\mu_n[Q_{Y^n|X^n}(f)>\eta]-c\ln\nu^{\otimes n}(f)
\nonumber
\\
	&\qquad\le n\,{\rm d}^{\star}(Q_X,Q_{Y|X},\nu,c)
	+\sqrt{6n\ln\frac{2}{\delta}}
	+c\sqrt{2n\ln\frac{1}{\eta}}
	+c\ln\frac{1}{\eta}
\label{e47}
\end{align}
for all $f\in \mathcal{H}_{[0,1]}(\mathbb{R}^n)$, $c>0$,
$\eta\in(0,1)$, where we defined
$\mu_n:=Q_X^{\otimes n}|_{\mathcal{C}_n}$.
\end{cor}

\begin{proof}
For $Q_X=\mathcal{N}(0,\sigma^2)$ and $Q_{Y|X=x}=\mathcal{N}(x,1)$, this 
follows immediately by combining Theorem \ref{thm_change_g}
and Theorem \ref{smooth:thm_gaussian} in Appendix \ref{app_dgaussian}.

For a general Gaussian measure $Q_{XY}$, we argue as follows.
Suppose first that
$Q_{Y|X=x}=\mathcal{N}(ax,a^2)$ for some $a\ne 0$.
We can reduce to the case already proved by applying
\eqref{e47} with $\tilde Y:=\frac{Y}{a}$, $\tilde f(y^n)=f(ay^n)$; as
$\nu^{\otimes n}(f)=|a|^n\nu^{\otimes n}(\tilde f)$ and
\begin{equation}
	{\rm d}^{\star}(Q_X,Q_{\tilde Y|X},\nu,c)=
	{\rm d}^{\star}(Q_X,Q_{Y|X},\nu,c)+c\ln|a|
\end{equation}
(the latter follows directly from Definition \ref{defn_dstar}),
it is readily verified that the 
conclusion of Corollary \ref{cor7} remains valid for any $a\ne 0$.
Now suppose that we further scale $X,Y$ simultaneously by some factor 
$b\ne 0$; then clearly both sides of \eqref{e47} remain unchanged.
As any nondegenerate Gaussian measure may be obtained by scaling
$X$ and $Y$ in this manner, the proof is complete.
\end{proof}

\subsection{Application: hypothesis testing with a communication constraint}
\label{sec_change_hypo}

We now revisit the hypothesis testing problem with communication 
constraint introduced in Section \ref{sec_change_ex}, and show 
how the general framework of Section \ref{sec_change_d} enables us to 
achieve a second-order converse in this setting. Let us note that the 
first-order term in the following result gives the correct rate for this 
problem: the matching achievability argument can be found in 
\cite{ahlswede1986hypothesis}.

\begin{thm}
\label{thm_change_hypo}
Let $|\mathcal{X}|<\infty$, and consider the hypothesis testing problem 
with communication constraint defined in Section \ref{sec_change_ex}.
Let $\mathcal{B}$ be any test of the form \eqref{eq_cc} that uses at most 
$M$ codewords and satisfies $\pi_{1|0}:=Q_{XY}^{\otimes n}[\mathcal{B}^c]\le
\epsilon$. Then the error probability $\pi_{0|1}:=Q_X^{\otimes 
n}Q_Y^{\otimes n}[\mathcal{B}]$ satisfies
\begin{align}
\label{eq_hypo_converse}
	\ln \pi_{0|1} 
	\ge &\mbox{} -n\sup_{U\colon U-X-Y}\left\{I(U;Y)\colon
	I(U;X)\le \frac{1}{n}\ln M\right\}
\\
\nonumber
	&
	-\bigg(
	2\ln(\alpha\beta)
	\sqrt{3\beta\ln\frac{4|\mathcal{X}|}{1-\epsilon}}
	+2\sqrt{\alpha\ln\frac{4}{1-\epsilon}}
	\bigg)\sqrt{n}
	-2\ln\frac{4}{1-\epsilon}
\end{align}
for $n>3\beta\ln\frac{4|\mathcal{X}|}{1-\epsilon}$,
where $\beta := \max_x \frac{1}{Q_X(x)}$ and 
$\alpha:=\max_x\left\|\frac{dQ_{Y|X=x}}{dQ_Y}\right\|_{\infty}$.
\end{thm}

\begin{proof}
We proceed in three steps.

\textbf{Step 1.} We begin with an expurgation argument. Suppose the test
$\mathcal{B}$ of the form \eqref{eq_cc} satisfies $\pi_{1|0}\le\epsilon$ 
and $\pi_{0|1}=\rho$.
We claim that for any $\epsilon'\in(0,1-\epsilon)$, we can modify the 
coding scheme such that the error under ${\rm H}_0$ is below 
$\epsilon+\epsilon'$, and the error under ${\rm H}_1$ is below 
$\frac{\rho}{\epsilon'}$ \emph{conditioned on each message}.

Indeed, assume without loss of generality that the messages $i=1,\ldots,M$ 
are ordered such that $Q_{X^n}Q_{Y^n}[\mathcal{B}|W=i]$ is increasing in 
$i$. Let
\begin{equation}
	i^\dagger := \min\{i:Q_{X^n}[W>i]\le\epsilon'\},
\end{equation}
and define a new test $\mathcal{B}'$ that always declares
${\rm H}_1$ upon receiving $i>i^\dagger$, and coincides with $\mathcal{B}$ 
otherwise. Then $\mathcal{B}'$ satisfies
\begin{equation}
	Q_{X^nY^n}[\mathcal{B}^{\prime c}] =
	Q_{X^nY^n}[\mathcal{B}^c\cap\{W\le i^\dagger\}]+
	Q_{X^n}[W>i^\dagger] \le \epsilon+\epsilon'.
\label{e456}
\end{equation}
On the other hand, as $Q_{X^n}Q_{Y^n}[\mathcal{B}'|W=i]=
Q_{X^n}Q_{Y^n}[\mathcal{B}|W=i]$ for $i\le i^\dagger$ and as we assumed 
this quantity is increasing in $i$, we can estimate
\begin{equation}
	\epsilon'Q_{X^n}Q_{Y^n}[\mathcal{B}'|W=i]
	\le
	Q_{X^n}Q_{Y^n}[\mathcal{B}\cap\{W\ge i^\dagger\}]
	\le \rho
\label{e457}
\end{equation}
for all $i=1,\ldots,M$ (this is trivial for $i>i^\dagger$ as
then $Q_{X^n}Q_{Y^n}[\mathcal{B}'|W=i]=0$).
Thus, the claimed properties are satisfied for the modified
test $\mathcal{B}'$.

\textbf{Step 2.} Our aim is to apply Corollary \ref{cor_alternative} to 
the modified test $\mathcal{B}'$. Fix for the time being any 
$\delta\in(0,1-\epsilon-\epsilon')$, and define $\mu_n$ as in Corollary 
\ref{cor_change_discrete}. Then
\begin{equation}
	1-\epsilon-\epsilon'
	\le \int Q_{X^nY^n}[\mathcal{B}'|X^n]\,dQ_{X^n}
	\le 
	\int Q_{X^nY^n}[\mathcal{B}'|X^n]\,
	d\mu_n + \delta
\end{equation}
by \eqref{e456}. Thus, we may bound, for any $c\ge 1$, 
\begin{equation}
	(1-\epsilon-\epsilon'-\delta)^{c(1+\frac{1}{t})}
	\le
	\int (Q_{X^nY^n}[\mathcal{B}'|X^n])^{c(1+\frac{1}{t})}\,
	d\mu_n,
\end{equation}
where we used Jensen's inequality and the fact that $\mu_n$ is a 
sub-probability measure.

Now denote by $\mathcal{B}_i'\subseteq\mathcal{Y}^n$ the set of sequences 
$y^n$ for which the test $\mathcal{B}'$ declares ${\rm H}_0$ upon 
receiving message $i$ (cf.\ \eqref{eq_cc}). Then we may crudely estimate
\begin{align}
\nonumber
	&(1-\epsilon-\epsilon'-\delta)^{c(1+\frac{1}{t})} 
\\ 	&\le 
	\sum_{i=1}^M 
	\int (Q_{Y^n|X^n}[\mathcal{B}_i'])^{c(1+\frac{1}{t})}\, d\mu_n
\\ 
	&\le 
	e^{n{\rm d}^\star(Q_X,Q_{Y|X},Q_Y,c) +
	c(\alpha-1)nt +
	\ln(\alpha^c\beta^{c+1})\sqrt{3n\beta\ln\frac{|\mathcal{X}|}{\delta}}}
	\sum_{i=1}^M
	(Q_{Y^n}[\mathcal{B}_i'])^c
\label{eblabli}
\end{align}
for $n>3\beta\ln\frac{|\mathcal{X}|}{\delta}$ and $t>0$, where we used 
Corollary \ref{cor_alternative} in \eqref{eblabli}. But as
$Q_{Y^n}[\mathcal{B}_i']\le\frac{\rho}{\epsilon'}$ by \eqref{e457},
we may rearrange the above estimate to obtain
\begin{align}
\nonumber
	\ln\rho
	\ge &\mbox{}
	-\frac{n}{c}\,{\rm d}^\star(Q_X,Q_{Y|X},Q_Y,c)
	-\frac{\ln M}{c} \\
\nonumber
&\mbox{}
	-(\alpha-1)nt 
	-\bigg(1+\frac{1}{t}\bigg)\ln\frac{1}{1-\epsilon-\epsilon'-\delta}
	\\
&\mbox{}
	-\frac{\ln(\alpha^c\beta^{c+1})}{c}
	\sqrt{3n\beta\ln\frac{|\mathcal{X}|}{\delta}}
	-\ln\frac{1}{\epsilon'}.
\end{align}
Finally, making the convenient choices
\begin{equation}
	\epsilon' = \frac{1-\epsilon}{2},\qquad\quad
	\delta = \frac{1-\epsilon}{4},
\end{equation}
recalling that $\rho=\pi_{0|1}$, and optimizing over $t>0$ yields
\begin{align}
\label{eq_pre_hypo}
	\ln\pi_{0|1}
	\ge &\mbox{}
	-\frac{n}{c}\,{\rm d}^\star(Q_X,Q_{Y|X},Q_Y,c)
	-\frac{\ln M}{c} \\
\nonumber
&\mbox{}
	-\bigg(
	2\ln(\alpha\beta)
	\sqrt{3\beta\ln\frac{4|\mathcal{X}|}{1-\epsilon}}
	+2\sqrt{(\alpha-1)\ln\frac{4}{1-\epsilon}}
	\bigg)\sqrt{n}
	-2\ln\frac{4}{1-\epsilon}
\end{align}
for $n>3\beta\ln\frac{4|\mathcal{X}|}{1-\epsilon}$ and $c\ge 1$.

\textbf{Step 3.} To achieve the correct first-order rate it remains to 
optimize \eqref{eq_pre_hypo} over $c\ge 1$. To this end, let us denote
the rate parameter that appears in \eqref{eq_hypo_converse} by
\begin{equation}
	\theta(R) := 
	\sup_{U\colon U-X-Y}\{I(U;Y)\colon
	I(U;X)\le R\}.
\end{equation}
When we choose $\nu=Q_Y$ in Definition \ref{defn_dstar}, the latter may be 
written as
\begin{align}
\label{eq_dstarmut}
	\frac{1}{c}{\rm d}^{\star}(Q_X, Q_{Y|X}, Q_Y, c) &=
	\sup_{U\colon U-X-Y}
	\left\{I(U;Y)-\frac{1}{c}I(U;X)\right\} \\
	&=
	\sup_{0\le R\le\ln|\mathcal{X}|}
	\left\{\theta(R)-\frac{R}{c}\right\},
\end{align}
where we used that $I(U;X)\le\ln|\mathcal{X}|$. But is not hard to show that
$\theta(R)$ is a concave function of $R$, cf.\ \cite[Lemma 
1]{ahlswede1986hypothesis}. Thus, we obtain for any $R'\ge 0$
\begin{align}
	&\inf_{c\ge 1}
	\bigg\{
	\frac{1}{c}{\rm d}^{\star}(Q_X, Q_{Y|X}, Q_Y, c) +
	\frac{R'}{c}
	\bigg\}
	=
	\\ &\inf_{c>0}
	\bigg\{
	\frac{1}{c}{\rm d}^{\star}(Q_X, Q_{Y|X}, Q_Y, c) +
	\frac{R'}{c}
	\bigg\}
	= \theta(R'),
\label{eblooper}
\end{align}
where we used in the first line that ${\rm 
d}^{\star}(Q_X, Q_{Y|X}, Q_Y, c)=0$ for $c\le 1$ by the data processing
inequality $I(U;Y)\le I(U;X)$; and we used the minimax theorem in 
\eqref{eblooper}. Setting $R'=\frac{1}{n}\ln M$ completes the proof.
\end{proof}

\subsection{Application: source coding with compressed side information}\label{sec_source}

In this section, we revisit one of the original applications of the 
blowing-up method: the problem of source coding with compressed side 
information \cite{wyner1975side,ahlswede_bounds_cond1976}. Using the 
methods of this section, we will obtain second-order converses for 
the discrete and Gaussian variants of this problem. The present setting is 
a typical example of \emph{side information problems}, whose converses 
have so far proved to be particularly elusive in non-asymptotic information 
theory (see, e.g., \cite[Section~9.2]{CIT-086}). Using the smoothing-out 
method, 
second-order converses can now be obtained by a straightforward 
replacement of the classical blowing-up analysis.
For second-order achievability bounds for side information problems, we 
refer the reader to \cite{watanabe2015}.

The problem that we investigate in this section is illustrated in 
Figure \ref{f_source}. Two correlated memoryless sources $X^n,Y^n$
with distribution $Q_{X^nY^n}=Q_{XY}^{\otimes n}$ are 
encoded into messages $W_1\in[M_1],W_2\in[M_2]$, respectively, and 
transmitted to a decoder. The (noiseless) communication channels are 
rate-constrained
in that they limit the sizes $M_1,M_2$ of the two codebooks. The decoder 
aims to reconstruct the source $Y^n$ with error 
probability at most $\epsilon$ (the side information $X^n$ is not 
reconstructed, but helps decode $Y^n$ due to the correlation between the 
sources). We aim to understand the fundamental interplay between 
$M_1,M_2$ and $\epsilon$.
\begin{figure}[t]
  \centering
\begin{tikzpicture}[node distance=0.6cm,minimum height=10mm,minimum width=12mm,arw/.style={->,>=stealth'}]
  \node[draw,rectangle] (D) {~~~Decoder~~~};
  \node[draw,rectangle] (E1) [left=1.4cm of D] {Encoder 2};
  \node[draw,rectangle] (E2) [below =0.5cm of D] {Encoder 1};
  \node[rectangle] (Y) [left =0.4cm of E1] {$Y^n$};
  \node[rectangle] (X) [left =0.4cm of E2] {$X^n$};
  \node[rectangle] (Yh) [right =0.4cm of D] {$\hat{Y}^n$};

  \draw[->]  (Y) to node[midway,above]{} (E1);
  \draw[->]  (X) to node[midway,right]{} (E2);
  \draw[->]  (E1) to node[midway,above]{$W_2$} (D);
  \draw[->]  (E2) to node[midway,right]{$W_1$} (D);
  \draw[->]  (D) to node[midway,left]{} (Yh);
\end{tikzpicture}
\caption{Source coding with compressed side information}
\label{f_source}
\end{figure}
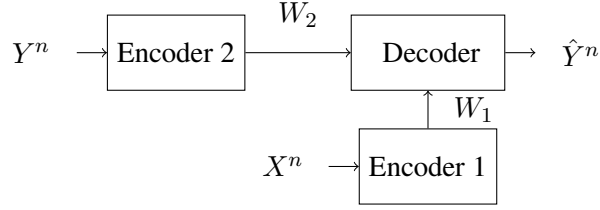

We begin by developing a second-order converse in the finite alphabet 
setting.

\begin{thm}\label{thm_sidediscrete}
Let $|\mathcal{X}|,|\mathcal{Y}|<\infty$, $\epsilon\in (0,1)$, and $n\ge 
3\beta_X\ln\frac{4|\mathcal{X}|}{1-\epsilon}$,
where $\beta_X$ is defined in \eqref{e52}.
Then for any encoders $f\colon \mathcal{X}^n\to[M_1]$,
$g\colon \mathcal{Y}^n\to [M_2]$, and decoder $\hat Y^n\colon 
[M_1]\times[M_2]\to \mathcal{Y}^n$ such that
$\mathbb{P}[Y^n\neq\hat{Y}^n] \le \epsilon$, we have
\begin{align}
\label{e_dsource}
	\ln M_2 
	&\ge 
	n \inf_{U\colon U-X-Y}\bigg\{H(Y|U)\colon 
		I(U;X)\le \frac{1}{n}\ln M_1\bigg\}
	\\
\nonumber
	&\quad
	- \bigg(
		2\ln (|\mathcal{Y}|\beta_X)
		\sqrt{3\beta_X\ln\frac{4|\mathcal{X}|}{1-\epsilon}}
		+ 2\sqrt{|\mathcal{Y}|\ln\frac{2}{1-\epsilon}}
	\bigg)\sqrt{n}
	- 2\ln\frac{4}{1-\epsilon}.
\end{align}
\end{thm}

Note that the first term on the right-hand side of \eqref{e_dsource} 
corresponds precisely to the rate region for the present problem
(see, e.g., \cite[Theorem~10.2]{el2011network}).

\begin{proof}
The proof follows from Corollary~\ref{cor_change_discrete} using similar
steps as in \cite[Theorem~3]{ahlswede_bounds_cond1976}.
Define for every $i\in[M_1]$ the set
\begin{align}
	\mathcal{B}_i := \{
	y^n\in\mathcal{Y}^n\colon
	y^n=\hat Y^n(i,g(y^n))\}
\end{align}
of correctly decoded sequences $y^n$, given that the side information
message $i$ was received. Then we have by assumption
\begin{align}
	Q_{X^n}[Q_{Y^n|X^n}[\mathcal{B}_{f(X^n)}]]\ge 1-\epsilon.
\end{align}
Thus, for any $\epsilon'\in (\epsilon,1)$, we obtain
\begin{align}
	Q_{X^n}[Q_{Y^n|X^n}[\mathcal{B}_{f(X^n)}]\ge 
	1-\epsilon']
	\ge 1-\frac{\epsilon}{\epsilon'}
\end{align}
by applying Markov's inequality to
$Q_{X^n}[1-Q_{Y^n|X^n}[\mathcal{B}_{f(X^n)}]>\epsilon']$.

Fix for the time being any $\delta\in(0,1-\frac{\epsilon}{\epsilon'})$ 
such that $n>3\beta_X\ln\frac{|\mathcal{X}|}{\delta}$, and define
$\mu_n$ as in Corollary~\ref{cor_change_discrete}. Then
\begin{align}
	\mu_n[Q_{Y^n|X^n}[\mathcal{B}_{f(X^n)}]\ge 
	1-\epsilon']
	\ge 1-\frac{\epsilon}{\epsilon'}-\delta.
\end{align}
As $f$ takes at most $M_1$ possible values, there exists $i^*$ such 
that
\begin{align}
	\mu_n[Q_{Y^n|X^n}[\mathcal{B}_{i^*}]\ge
        1-\epsilon'] \ge
	\frac{1-\frac{\epsilon}{\epsilon'}-\delta}{M_1}
\end{align}
by the union bound. On the other hand, if we let $\nu$ be the uniform 
distribution on $\mathcal{Y}$, then we can estimate by the definition of 
$\mathcal{B}_i$
\begin{align}
	\nu^{\otimes n}[\mathcal{B}_{i}]
	=|\mathcal{Y}|^{-n}|\mathcal{B}_{i}|\le
	M_2|\mathcal{Y}|^{-n}
\end{align}
for every $i$. Applying Corollary~\ref{cor_change_discrete} with
$f=1_{\mathcal{B}_{i^*}}$, $\eta=1-\epsilon'$ yields 
\begin{align}
\nonumber
	&\ln\frac{1-\frac{\epsilon}{\epsilon'}-\delta}{M_1}
	- c\ln M_2|\mathcal{Y}|^{-n} \\
	&
	\le
	\ln \mu_n[Q_{Y^n|X^n}[\mathcal{B}_{i^*}]\ge 1-\epsilon']
	- c\ln\nu^{\otimes n}[\mathcal{B}_{i^*}]
	\\
	&
	\le
	n\,{\rm d}^\star(Q_X,Q_{Y|X},\nu,c) + A\sqrt{n}
	+ c\ln\frac{1}{1-\epsilon'},
\end{align}
where $A$ is defined in \eqref{eq_cor_chg_disc_A}. Rearranging yields
\begin{align}
\nonumber
	\ln M_1 + c\ln M_2 
	\ge &\mbox{}
	-n\{{\rm d}^\star(Q_X,Q_{Y|X},\nu,c)-c\ln|\mathcal{Y}|\}
	\\
\nonumber
	&\mbox{}
	- \sqrt{n}\bigg(
	\ln (|\mathcal{Y}|^c\beta_X^{c+1})
	\sqrt{3\beta_X\ln\frac{4|\mathcal{X}|}{1-\epsilon}}
	+
	2c\sqrt{|\mathcal{Y}|\ln\frac{2}{1-\epsilon}}
	\bigg)
	\\
	&\mbox{}
	- c\ln\frac{2}{1-\epsilon}
	- \ln\frac{4}{1-\epsilon},
\label{eq_source_raw}
\end{align}
for $n>3\beta_X\ln\frac{4|\mathcal{X}|}{1-\epsilon}$,
where we made the choices
$\epsilon'=\frac{1+\epsilon}{2}$ and
$\delta=\frac{1}{2}(1-\frac{\epsilon}{\epsilon'})$.

It remains to manipulate the first-order terms. Note first that
\begin{align}
	{\rm d}^\star(Q_X,Q_{Y|X},\nu,c)-c\ln|\mathcal{Y}|
	&= 
	\sup_{U\colon U-X-Y}\{-cH(Y|U)-I(U;X)\} \\
	&=
	{\rm d}^\star(Q_X,Q_{Y|X},Q_Y,c) - cH(Q_Y),
\end{align}
which follows readily from Definition \ref{defn_dstar} and
\eqref{eq_dstarmut}. We may therefore follow verbatim the arguments
in Step 3 of the proof of Theorem \ref{thm_change_hypo} to optimize
\eqref{eq_source_raw} over $c\ge 1$, which readily completes the proof of
\eqref{e_dsource}.
\end{proof}

We now give a analogue of Theorem~\ref{thm_sidediscrete} for Gaussian 
sources under quadratic distortion.
We note again that the first-order term of \eqref{e_gsource} corresponds 
to the known rate region for this problem (e.g., let $D_2\to\infty$ in 
\cite[Theorem~12.3]{el2011network}).

\begin{thm}\label{thm_sidegaussian}
Let $\mathcal{X}=\mathcal{Y}=\mathbb{R}$,
$Q_{XY}=\mathcal{N}(0,\big(\begin{smallmatrix} 1 & \rho \\ \rho & 1
\end{smallmatrix}\big))$, and $D>0$, $\epsilon\in (0,1)$,
$n\ge 20\ln\frac{8}{1-\epsilon}$. Then for any encoders
$f\colon \mathbb{R}^n\to[M_1]$, $g\colon \mathbb{R}^n\to 
[M_2]$, and decoder $\hat Y^n\colon [M_1]\times[M_2]\to \mathbb{R}^n$
such that $\mathbb{P}[\|Y^n-\hat{Y}^n\|^2>nD] \le \epsilon$, we have
\begin{align}
\label{e_gsource}
	\ln M_2
	\ge 
	\frac{n}{2} \ln\frac{1-\rho^2+\rho^2e^{-\frac{2}{n}\ln M_1}}{D}
	-4\sqrt{\ln\frac{8}{1-\epsilon}}\sqrt{n}
	-2\ln\frac{4}{1-\epsilon}.
\end{align}
\end{thm}

\begin{proof}
Define for every $i\in[M_1]$ the set
\begin{align}
	\mathcal{B}_i := \{
	y^n\in\mathbb{R}^n\colon
	\|y^n-\hat Y^n(i,g(y^n))\|^2\le nD
	\}
\end{align}
of correctly decoded sequences $y^n$, given that the side information 
message $i$ was received. Let $\epsilon'\in(\epsilon,1)$, 
$\delta\in(0,1-\frac{\epsilon}{\epsilon'})$ such that $n\ge 
20\ln\frac{2}{\delta}$, define $\mu_n$ as in Corollary \ref{cor7}, and let 
$\nu$ be the Lebesgue measure on $\mathbb{R}$.
Repeating verbatim the arguments in the proof of Theorem 
\ref{thm_sidediscrete}, we find that there exists $i^*$ such that
\begin{align}
	\mu_n[Q_{Y^n|X^n}[\mathcal{B}_{i^*}]\ge
        1-\epsilon'] \ge
	\frac{1-\frac{\epsilon}{\epsilon'}-\delta}{M_1}.
\end{align}
On the other hand, in the present setting, we estimate by the union bound
\begin{align}
	\nu^{\otimes n}[\mathcal{B}_i] \le
	M_2\mathrm{Vol}(\mathrm{Ball}(0,\sqrt{nD}))
	\le M_2(2\pi e D)^{n/2}.
\end{align}
Applying Corollary \ref{cor7} with 
$f=1_{\mathcal{B}_{i^*}}$, $\eta=1-\epsilon'$ and rearranging yields
\begin{align}
\nonumber
	\ln M_1+c\ln M_2
	\ge &\mbox{}
	-n\{{\rm d}^{\star}(Q_X,Q_{Y|X},\nu,c)
	+\tfrac{c}{2}\ln(2\pi e D)\}
	\\
\nonumber
	&\mbox{}
	-\sqrt{n}\bigg(\sqrt{6\ln\frac{8}{1-\epsilon}}
	+c\sqrt{2\ln\frac{2}{1-\epsilon}}\bigg)
	\\
	&\mbox{}
	-c\ln\frac{2}{1-\epsilon}
	-\ln\frac{4}{1-\epsilon}
\label{eq_sidegconv}
\end{align}
for $n\ge 20\ln\frac{8}{1-\epsilon}$,
where we made the choices
$\epsilon'=\frac{1+\epsilon}{2}$ and
$\delta=\frac{1}{2}(1-\frac{\epsilon}{\epsilon'})$.

It remains to manipulate the first-order terms. To this end, we first
note that using \eqref{eq_dstargauss} in
Appendix \ref{app_dgaussian} and the subsequent discussion, we may compute
\begin{align}
\nonumber
	&{\rm d}^{\star}(Q_X,Q_{Y|X},\nu,c) + \tfrac{c}{2}\ln(2\pi e D)
	\\
	&=
	\sup_{\sigma\in[0,1]}\bigg\{\frac{1}{2}\ln \sigma^2
	-\frac{c}{2}\ln\frac{1-\rho^2+\rho^2\sigma^2}{D}\bigg\}
\label{eq_insidesup}
	\\
	&=
	\frac{c}{2}\, \vartheta^*\bigg(-\frac{1}{c}\bigg),
\end{align}
where $\vartheta^*$ is the convex conjugate of the convex function 
$\vartheta(x) := \ln\frac{1-\rho^2+\rho^2e^{-x}}{D}$ for $x\ge 0$ and
$\vartheta(x) := +\infty$ for $x<0$. Therefore, by convex duality
\begin{align}
\label{eq_insideinf}
	&\inf_{c>0}
	\frac{1}{c}\bigg\{
	{\rm d}^{\star}(Q_X,Q_{Y|X},\nu,c) + \tfrac{c}{2}\ln(2\pi e D)
	+ \frac{1}{n}\ln M_1
	\bigg\} 
	\\
	&= - \frac{1}{2} \vartheta\bigg(\frac{2}{n}\ln M_1\bigg)
	= - \frac{1}{2} \ln\frac{1-\rho^2+\rho^2e^{-\frac{2}{n}\ln M_1}}{D}.
\label{eq_sidegrate}
\end{align}
We claim that the same conclusion follows if we take the infimum 
over $c\ge 1$ only. Indeed, it is readily verified that the function inside 
the supremum in \eqref{eq_insidesup} is increasing in $\sigma^2$ for any
$c<1$, so that ${\rm d}^{\star}(Q_X,Q_{Y|X},\nu,c) + \tfrac{c}{2}\ln(2\pi e 
D) = -\frac{c}{2}\ln\frac{1}{D}$ for $c<1$. Thus, the infimum in 
\eqref{eq_insideinf} cannot be attained for $c<1$. The proof is now readily 
completed using \eqref{eq_sidegrate} by optimizing
\eqref{eq_sidegconv} over $c\ge 1$.
\end{proof}

\section*{Acknowledgements}

Suggestions on the presentation of the paper by Yury Polyanskiy are
greatly appreciated. We thank Wei Yang for discussions on multiple
access channels.

\appendix

\section{Data processing cannot yield second-order converses}

\label{app_data}

The aim of this appendix is to explain the claim made in Section 
\ref{sec_bul} that second-order converses are fundamentally outside 
the reach of any method that is based on the data processing inequality: 
in particular, the inefficiency of the blowing-up method already appears 
in the very first step in the argument. This claim is little more than the 
Neyman-Pearson lemma in the following disguise.

\begin{lem}
\label{lem_impo}
Let $P,Q$ be probability measures on $\mathcal{Y}$ such that
$\big\|\ln\frac{dP}{dQ}\big\|_\infty<\infty$.
Then for every $\delta>0$, there exists $C_\delta>0$ independent
of $n$ so that any set
$\mathcal{\tilde{A}}\subseteq\mathcal{Y}^n$ with $P^{\otimes
n}[\mathcal{\tilde{A}}]\ge 1-n^{-\delta}$ must satisfy $\ln Q^{\otimes
n}[\mathcal{\tilde{A}}]\ge -nD(P\|Q)+C_\delta\sqrt{n\ln
n}$.
\end{lem}

Let us first explain why this lemma rules out obtaining sharp bounds by
data processing. Suppose we start our argument by applying the data processing
inquality \eqref{uz} to an arbitrary set $\mathcal{\tilde A}$. Then we
obtain (cf.\ \eqref{eq_weakblup})
\begin{equation}
	\ln Q^{\otimes n}[\mathcal{\tilde A}] \ge
	-n\frac{D(P\|Q)}{P^{\otimes n}[\mathcal{\tilde A}]} - \ln 2.
\end{equation}
In order to deduce from this a non-asymptotic converse with $\sim\sqrt{n}$
second-order term, we must introduce some operation
$\mathcal{A}\mapsto\mathcal{\tilde A}$ that associates to \emph{any}
deterministic test
$\mathcal{A}$ with $P^{\otimes n}[\mathcal{A}]\ge 1-\epsilon$ a
set $\mathcal{\tilde A}$ with the following properties:
\begin{equation}
	P^{\otimes n}[\mathcal{\tilde A}]\ge 1-O(n^{-1/2}),\qquad
	Q^{\otimes n}[\mathcal{\tilde A}]\le e^{O(\sqrt{n})}
	Q^{\otimes n}[\mathcal{A}].
\end{equation}
But Lemma \ref{lem_impo} shows such an operation cannot exist. Indeed, let
$\mathcal{A}$ be a test with $P^{\otimes n}[\mathcal{A}]=1-\epsilon$ and
$\ln Q^{\otimes n}[\mathcal{A}]\le -nD(P\|Q) + O(\sqrt{n})$ as in
Lemma \ref{lem_bhcach}. Then by Lemma \ref{lem_impo}, any set
$\mathcal{\tilde A}$ such that $P^{\otimes n}[\mathcal{\tilde A}]\ge
1-O(n^{-1/2})$ satisfies
\begin{equation}
	Q^{\otimes n}[\mathcal{\tilde A}] \ge
	e^{-nD(P\|Q)+C\sqrt{n\ln n}}
	\ge
	e^{C\sqrt{n\ln n}+O(\sqrt{n})}Q^{\otimes n}[\mathcal{A}],
\end{equation}
which contradicts the desired property of the mapping
$\mathcal{A}\mapsto\mathcal{\tilde A}$. The same argument shows
that converses obtained from the data processing inequality can
attain at best a second-order term of order no smaller than
$\sim\sqrt{n\log n}$.

We conclude this appendix with the proof of Lemma \ref{lem_impo}.

\begin{proof}[Proof of Lemma \ref{lem_impo}]
By the Neyman-Pearson lemma, among all sets $\mathcal{\tilde A}$ such that
$P^{\otimes n}[\mathcal{\tilde A}]\ge 1-n^{-\delta}$, the probability
$Q^{\otimes n}[\mathcal{\tilde A}]$ is minimized by
\begin{equation}
	\mathcal{\tilde A} = \{y^n\in\mathcal{Y}^n:
	\imath_{P^{\otimes n}\|Q^{\otimes n}}(y^n) \ge
	nD(P\|Q) - \gamma\},
\end{equation}
where $\gamma$ is chosen so that $P^{\otimes n}[\mathcal{\tilde A}]=
1-n^{-\delta}$. It therefore suffices to prove a suitable lower bound on
$Q^{\otimes n}[\mathcal{\tilde A}]$ for this particular choice of
$\mathcal{\tilde A}$.

We will twice use the classical Kolmogorov lower bound for the tail
of sums of independent random variables \cite[p.\ 266]{pollard}.
First, this bound implies
\begin{equation}
\label{eq_kolmo}
	P^{\otimes n}[\imath_{P^{\otimes n}\|Q^{\otimes n}}\ge
	nD(P\|Q)-\sqrt{\delta V(P\|Q)}\sqrt{n\ln n}]
	\le 1-n^{-\delta}
\end{equation}
for $n$ sufficiently large (depending on
$\delta,V(P\|Q),\|\ln\frac{dP}{dQ}\|_\infty$), where
$V(P\|Q):=\mathrm{Var}_P(\ln\frac{dP}{dQ})$. This shows that we must have
$\gamma \ge \gamma_n := \sqrt{\delta V(P\|Q)}\sqrt{n\ln n}$.

On the other hand, we can now estimate
\begin{align}
	Q^{\otimes n}[\mathcal{\tilde A}] &=
	P^{\otimes n}\big(e^{-\imath_{P^{\otimes n}\|Q^{\otimes n}}}
	1_{\{\imath_{P^{\otimes n}\|Q^{\otimes n}}\ge
        nD(P\|Q)-\gamma\}}\big)
	\\
	&\ge
	\nonumber
	e^{-nD(P\|Q)+\sqrt{\delta V(P\|Q)/4}\sqrt{n\ln n}}\times
	\\
	&\qquad
	P^{\otimes n}[nD(P\|Q)-\tfrac{1}{2}\gamma_n \ge \imath_{P^{\otimes n}\|Q^{\otimes n}}
		\ge nD(P\|Q)-\gamma].
\end{align}
But applying again \eqref{eq_kolmo}, we find that
\begin{align}
	&
	P^{\otimes n}[nD(P\|Q)-\tfrac{1}{2}\gamma_n \ge \imath_{P^{\otimes n}\|Q^{\otimes n}}
		\ge nD(P\|Q)-\gamma]\\
	&\ge
	1-n^{-\delta}
	- P^{\otimes n}[
	\imath_{P^{\otimes n}\|Q^{\otimes n}} \ge nD(P\|Q)-\tfrac{1}{2}\gamma_n]
	\ge
	n^{-\delta/4}-n^{-\delta}.	
\end{align}
Consequently $\ln Q^{\otimes n}[\mathcal{\tilde A}]
\ge -nD(P\|Q) + \sqrt{\delta V(P\|Q)/4}\sqrt{n\ln n} + O(\ln n)$.
\end{proof}

\section{Brascamp-Lieb divergence: auxiliary results}

\label{app_proofs}

\subsection{Proof of \eqref{e_49} in the discrete case}\label{app_sbl}

The aim of this section is to prove the following result.

\begin{thm}\label{thm_sbl}
Let $|\mathcal{X}|<\infty$, $Q_X$ be a probability measure on 
$\mathcal{X}$, $\nu$ be a probability measure on $\mathcal{Y}$,
and $Q_{Y|X}$ be a random transformation.
Define
\begin{align}
	\beta_X & := \frac{1}{\min_x Q_X(x)}, \\
	\alpha_Y & :=
	\left\|\frac{dQ_Y}{d\nu}\right\|_{\infty}.
\end{align}
Then for every $\delta\in (0,1)$ and 
$n>3\beta_X\ln\frac{|\mathcal{X}|}{\delta}$,
we may choose a set $\mathcal{C}_n\subseteq\mathcal{X}^n$ with
$Q_X^{\otimes n}[\mathcal{C}_n]\ge 1-\delta$ such that 
\begin{align}
\nonumber
	&{\rm d}(\mu_n,Q_{Y|X}^{\otimes n},\nu^{\otimes n},c)
	\\ &\qquad\le
	n\, {\rm d}^{\star}(Q_X,Q_{Y|X},\nu,c)
	+\ln (\alpha_Y^c\beta_X^{c+1})
	\sqrt{3n\beta_X\ln\frac{|\mathcal{X}|}{\delta}}
\label{e147}
\end{align}
for every $c>0$, where we defined $\mu_n:=Q_X^{\otimes n}|_{\mathcal{C}_n}$. 
\end{thm}

Let us fix in the sequel $(Q_X,Q_{Y|X},\nu)$ and $c>0$, and define
\begin{equation}
	\phi(P_X) :=
	{\rm d}^{\star}(P_X,Q_{Y|X},\nu,c)
\label{eq_dstarphi}
\end{equation}
for any $P_X\ll Q_X$.
The idea behind the proof Theorem \ref{thm_sbl} is roughly as follows. 
Using the chain rule of relative entropy, we will show that ${\rm 
d}(\mu_n,Q_{Y|X}^{\otimes n},\nu^{\otimes n},c)$ can be bounded by
a quantity of the form $n\phi(P_X)$, where $P_X$ is a mixture of empirical 
measures of sequences in the support of $\mu_n$. 
Thus, if we choose $\mathcal{C}_n$ to be the set of sequences with 
empirical measure close to $Q_X$, then $Q_X^{\otimes 
n}[\mathcal{C}_n]$ will be large by the law of large numbers and 
$\phi(P_X)\approx\phi(Q_X)={\rm d}^{\star}(Q_X,Q_{Y|X},\nu,c)$,
completing the proof.

To make these ideas precise, we require quantitative forms of the 
continuity of $\phi$ and of the law of large numbers. The former is 
provided by the following lemma. Here we adopt the same notations as
in Theorem \ref{thm_sbl}.

\begin{lem}\label{lem_us}
If $P_X\le (1+\epsilon)Q_X$ for some $\epsilon\in [0,1)$, then
\begin{align}
	\phi(P_X)\le \phi(Q_X)+
	\epsilon \ln(\beta_X^{c+1}\alpha_Y^c).
\end{align}
\end{lem}

\begin{proof}
Let $P_X$ be as in the statement, and let $P_{U|X}$ be a maximizer for 
\eqref{eq_dstar} (we assume its existence for notational simplicity only;
if a maximizer does not exist, the argument is readily adapted to work
with near-maximizers). We now modify this distribution by allowing $U$ to 
take an additional value as follows: let
$\tilde{\mathcal{U}}=\mathcal{U}\cup\{\star\}$,
$\tilde{\mathcal{X}}=\mathcal{X}$, and define
$P_{\tilde{U}\tilde{X}}$ by
\begin{align}
	P_{\tilde{U}}
	&:=\frac{1}{1+\epsilon}P_U+\frac{\epsilon}{1+\epsilon}\delta_{\star};
\\
	P_{\tilde{X}|\tilde{U}=u}&:=P_{X|U=u},\quad \forall u\in\mathcal{U};
\\
	P_{\tilde{X}|\tilde{U}=\star}&:=
	\frac{1+\epsilon}{\epsilon}\left(Q_X-\frac{1}{1+\epsilon}P_X\right)
\end{align}
where $\delta_{\star}$ is a point mass on $\star$.
Observe that $P_{\tilde{X}}=Q_X$,
\begin{align}
	D(P_{\tilde{X}|\tilde{U}}\|Q_X|P_{\tilde{U}})
	&=\frac{1}{1+\epsilon}D(P_{X|U}\|Q_X|P_U)
\nonumber
\\
	&\quad
	+\frac{\epsilon}{1+\epsilon}D
	\left(\left.\frac{1+\epsilon}{\epsilon}Q_X-\frac{1}{\epsilon}P_X
	\right\|Q_X\right)
\\
	&\le D(P_{X|U}\|Q_X|P_U) + \epsilon\ln\beta_X
\label{eq_phi1}
\end{align}
where we used $D(P\|Q_X)\le \ln\beta_X$ for every probability measure 
$P$, and
\begin{align}
	cD(P_{\tilde{Y}|\tilde{U}}\|\nu|P_{\tilde{U}})
	&=\frac{c}{1+\epsilon}D(P_{Y|U}\|\nu|P_U)
\nonumber
\\
	&\quad+\frac{c\epsilon}{1+\epsilon}D
	\left(\left.\frac{1+\epsilon}
	{\epsilon}Q_Y-\frac{1}{\epsilon}P_Y
	\right\|\nu\right) \\
	&\ge
	cD(P_{Y|U}\|\nu|P_U) -
	c\epsilon \ln(\beta_X\alpha_Y),
\label{eq_phi2}
\end{align}
where we used $\frac{1}{1+\epsilon}\ge 1-\epsilon$ and
\begin{align}
	D(P_{Y|U}\|\nu|P_U)
	&=
	D(P_{Y|U}\|Q_Y|P_U)
	+ 
	P_Y(\ln\tfrac{dQ_Y}{d\nu})
	\\
	&\le 
	D(P_{X|U}\|Q_X|P_U)
	+P_Y(\ln\tfrac{dQ_Y}{d\nu})
	\label{e_b14}	
	\\
	&\le
	\ln(\beta_X\alpha_Y),
	\label{e_b15}
\end{align}
where \eqref{e_b14} follows from the data processing inequality and \eqref{e_b15} follows since $D(P_{X|U}\|Q_X|P_U)\le\ln\beta_X$. 
The proof is concluded by subtracting 
\eqref{eq_phi1} from \eqref{eq_phi2}.
\end{proof}

As a quantitative form of the law of large numbers, we will use the 
following standard result (see, e.g., \cite{boucheron2013}).

\begin{lem}[Chernoff Bound for Bernoulli variables]
\label{lem_BE}
Assume that $X_1,\dots,X_n$ are i.i.d.~${\rm Ber}(p)$.
Then for any $\epsilon\in (0,\infty)$
\begin{align}
	\mathbb{P}\left[\sum_{i=1}^nX_i\ge (1+\epsilon)np\right]
	\le e^{-\frac{1}{3}\min\{\epsilon^2,\epsilon\}np}.
\end{align}
\end{lem}

We can now conclude the proof of Theorem \ref{thm_sbl}.

\begin{proof}[Proof of Theorem~\ref{thm_sbl}]
We denote by $\widehat{P}_{X^n}$ the empirical measure of $X^n\sim 
Q_X^{\otimes n}$. Let $n>3\beta_X\ln\frac{|\mathcal{X}|}{\delta}$ and 
define
\begin{align}
	\epsilon_n &:=
	\sqrt{\frac{3\beta_X}{n}\ln\frac{|\mathcal{X}|}{\delta}}
	\in
	(0,1),
\label{e_epn}
	\\
	\mathcal{C}_n &:=\{x^n\colon \hat{P}_{x^n}\le (1+\epsilon_n)Q_X\}.
\label{e_104}
\end{align}
As for each $x\in\mathcal{X}$
\begin{align}
	\mathbb{P}[\widehat{P}_{X^n}(x)> (1+\epsilon_n)Q_X(x)]
	\le
	e^{-\frac{n}{3}Q_X(x)\epsilon_n^2} 
	\le
	\frac{\delta}{|\mathcal{X}|}
\end{align}
by Lemma~\ref{lem_BE}, it follows by the union bound that
$Q_X^{\otimes n}[\mathcal{C}_n]\ge 1-\delta$.

Now consider any probability measure $P_{X^n}\ll\mu_n:=Q_X^{\otimes 
n}|_{\mathcal{C}_n}$. Then we can estimate, following essentially
\cite[Lemma~9]{ISIT_lccv_smooth2016},
\begin{align}
	&cD(P_{Y^n}\|\nu^{\otimes n})-D(P_{X^n}\|\mu_n)
\nonumber\\
&=
	cD(P_{Y^n}\|\nu^{\otimes n})-D(P_{X^n}\|Q_X^{\otimes n})
\label{e115}
\\
&=
	c\sum_{i=1}^{n}D(P_{Y_i|Y^{i-1}}\|\nu|P_{Y^{i-1}})
	-\sum_{i=1}^{n}D(P_{X_i|X^{i-1}}\|Q_X|P_{X^{i-1}})
\label{e116bis}
\\
&\le
	c\sum_{i=1}^{n}D(P_{Y_i|X^{i-1}}\|\nu|P_{X^{i-1}})
	-\sum_{i=1}^{n}D(P_{X_i|X^{i-1}}\|Q_X|P_{X^{i-1}})
\label{e117}
\\
	&=n[c\,D(P_{Y_I|IX^{I-1}}\|\nu|P_{IX^{I-1}})
	-D(P_{X_I|IX^{I-1}}\|Q_X|P_{IX^{I-1}})]
\label{e118}
\\
&\le
	n\phi(P_{X_I}).
\label{e105}
\end{align}
\eqref{e115} follows from the fact that $P_{X^n}$ is supported on
$\mathcal{C}_n$; \eqref{e116bis} is the chain rule of relative entropy;
\eqref{e117} follows from the convexity of relative entropy since 
$Y_i-X^{i-1}-Y^{i-1}$ under $P_{X^nY^n}$; in \eqref{e118} we defined $I$ to 
be a random variable uniformly distributed on $\{1,\dots,n\}$ and 
independent of $X^n,Y^n$; and in \eqref{e105} $\phi$ is defined in 
\eqref{eq_dstarphi}. But note that $P_{X_I} = P_{X^n}(\hat P_{X^n})\le 
(1+\epsilon_n)Q_X$ as $P_{X^n}$ is supported on $\mathcal{C}_n$.
Thus, \eqref{e105}, Lemma~\ref{lem_us} and \eqref{eq_dstarphi} yield
\begin{align}
	{\rm d}(\mu_n,Q_{Y|X}^{\otimes n},\nu^{\otimes n},c)
	&=
	\sup_{P_{X^n}\ll\mu_n}\{cD(P_{Y^n}\|\nu^{\otimes n})-D(P_{X^n}\|\mu_n)\}
	\\
	&
	\le
	n{\rm d}^{\star}(Q_X,Q_{Y|X},\nu,c) +
	\ln(\beta_X^{c+1}\alpha_Y^c)n\epsilon_n,
\end{align}
and the proof is complete.
\end{proof}

\subsection{Proof of \eqref{e_49} in the Gaussian case}
\label{app_dgaussian}

We now prove the following Gaussian analogue of Theorem \ref{thm_sbl}.

\begin{thm}\label{smooth:thm_gaussian}
Let $Q_X=\mathcal{N}(0,\sigma^2)$, let $Q_{Y|X=x}=\mathcal{N}(x,1)$, and 
let $\nu$ be the Lebesgue measure on $\mathbb{R}$.
Then for any $\delta\in(0,1)$ and $n\ge 20\ln\frac{2}{\delta}$, we may 
choose a set $\mathcal{C}_n\subseteq\mathbb{R}^n$ with
$Q_X^{\otimes n}[\mathcal{C}_n]\ge 1-\delta$ such that
\begin{align}
	{\rm d}(\mu_n,Q_{Y|X}^{\otimes n},\nu^{\otimes n},c)\le
	n{\rm d}^{\star}(Q_X,Q_{Y|X},\nu,c)
	+\sqrt{6n\ln\frac{2}{\delta}}
\label{smooth:e55}
\end{align}
for any $c>0$, where we defined $\mu_n:=Q_X^{\otimes n}|_{\mathcal{C}_n}$.
\end{thm}

In the Gaussian case, the choice of typical set $\mathcal{C}_n$ is rather 
simple: it is a spherical shell of radius $\sim\sqrt{n}$ and width 
$O(1)$. Controlling the probability of such a spherical shell is a 
standard exercise (cf.\ \cite[p.\ 43]{boucheron2013}).

\begin{lem}[Chi-square tail bound]
\label{lem_chisq}
For $X^n\sim\mathcal{N}(0,\mathbf{I}_n)$
and any $t>0$
\begin{equation}
	\mathbb{P}[|\|X^n\|^2-n|\ge 2\sqrt{nt}+2t]\le 
	2e^{-t}.
\end{equation}
\end{lem}

We now turn to the proof of Theorem \ref{smooth:thm_gaussian}.

\begin{proof}[Proof of Theorem \ref{smooth:thm_gaussian}]
Define
\begin{align}
	A &:= \sqrt{6\ln\frac{2}{\delta}},\\
	\mathcal{C}_n &:=
	\{x^n: n-A\sqrt{n}\le \tfrac{1}{\sigma^2} \|x^n\|^2 \le 
	n+A\sqrt{n}\}.
\end{align}
Then we have $Q_X^n[\mathcal{C}_n]\ge 1-\delta$ for
$n\ge 20\ln\frac{2}{\delta}$ by Lemma \ref{lem_chisq}.

We proceed similarly to the proof of Theorem \ref{thm_sbl} 
(following essentially \cite[Theorem~13]{ISIT_lccv_smooth2016}).
Consider any probability measure 
$P_{X^n}\ll\mu_n:=Q_X^{\otimes n}|_{\mathcal{C}_n}$. Then
\begin{align}
\nonumber
	&cD(P_{Y^n}\|\nu^{\otimes n}) - D(P_{X^n}\|\mu_n) \\
\label{e0e0}
	&= cD(P_{Y^n}\|\nu^{\otimes n}) - D(P_{X^n}\|Q_X^{\otimes n}) \\
\label{e0e1}
	&= cD(P_{Y^n}\|\nu^{\otimes n}) - D(P_{X^n}\|\nu^{\otimes n})
		+ P_{X^n}(\ln\tfrac{dQ_X^{\otimes n}}{d\nu^{\otimes n}}) \\
\label{e0e2}
	&\le h(P_{X^n}) - ch(P_{Y^n}) 
		- nh(Q_X)
		+ \tfrac{1}{2}A\sqrt{n},
\end{align}
where in \eqref{e0e0} we used that $P_{X^n}$ is supported on 
$\mathcal{C}_n$; \eqref{e0e1} follows from the definition of relative
entropy; and \eqref{e0e2} follows by substituting the explicit form of the 
Gaussian density $\frac{dQ_X}{d\nu}$ and using the definition of 
$\mathcal{C}_n$.

We now proceed to estimate using the chain rule
\begin{align}
\nonumber
	&h(P_{X^n}) - ch(P_{Y^n}) \\
	&=\sum_{i=1}^n \{h(X_i|X^{i-1}) - ch(Y_i|Y^{i-1})\} \\
\label{e1e0}
	&\le\sum_{i=1}^n \{h(X_i|X^{i-1}) - ch(Y_i|X^{i-1})\} \\
\label{e1e1}
	& = n\{h(X_I|IX^{I-1}) - ch(Y_I|IX^{I-1})\} \\
\label{e1e2}
	&\le nF((1+\tfrac{A}{\sqrt{n}})\sigma^2),
\end{align}
where in \eqref{e1e0} we used concavity of differential entropy;
in \eqref{e1e1} we defined $I$ to be uniformly distributed on 
$\{1,\ldots,n\}$ and independent of everything else; and in
\eqref{e1e2} we used that $\mathrm{Var}(X_I)\le 
(1+\tfrac{A}{\sqrt{n}})\sigma^2$
by the definition of $\mathcal{C}_n$, with
\begin{equation}
\label{eq_fmfm}
	F(M) := \sup_{P_{UX}\colon \mathrm{Var}(X)\le M}
	\{h(X|U)-c\,h(Y|U)\}
\end{equation}
(this quantity plays the role of $\phi(P_X)$ in the present setting).
But note that by the scaling property of the differential entropy, we 
readily obtain
\begin{equation}
\label{e2e2}
	nF((1+\tfrac{A}{\sqrt{n}})\sigma^2) =
	nF(\sigma^2) + \tfrac{1-c}{2}n\ln(1+\tfrac{A}{\sqrt{n}})
	\le 
	nF(\sigma^2)+
	\tfrac{1}{2}A\sqrt{n}.
\end{equation}
Combining \eqref{e0e2}, \eqref{e1e2}, and \eqref{e2e2}, we have shown
\begin{align}
	{\rm d}(\mu_n,Q_{Y|X}^{\otimes n},\nu^{\otimes n},c)
	&=
	\sup_{P_{X^n}\ll\mu_n}\{
	cD(P_{Y^n}\|\nu^{\otimes n}) - D(P_{X^n}\|\mu_n)\} \\
	&\le
	nF(\sigma^2) - nh(Q_X) + A\sqrt{n}.
\end{align}
To conclude the proof, it remains to show that
\begin{equation}
	F(\sigma^2) - h(Q_X) =
	{\rm d}^\star(Q_X,Q_{Y|X},\nu,c).
\label{eq_finalequationofthepaper}
\end{equation}
To this end, we note that arguing as in \eqref{e0e1}, we may
write Definition \ref{defn_dstar} as
\begin{equation}
\label{eq_dstargauss}
	{\rm d}^\star(Q_X,Q_{Y|X},\nu,c) =
	\sup_{P_{UX}\colon P_X=Q_X}
	\{h(X|U)-c\,h(Y|U)\} - h(Q_X).
\end{equation}
Thus, the inequality $\ge$ in \eqref{eq_finalequationofthepaper}
is immediate from the definition \eqref{eq_fmfm}. For the converse
direction we require the fact, proved in \cite[Theorem 14]{lccv2015}, that 
for Gaussian channels the supremum in \eqref{eq_fmfm} is achieved for
$P_{UX}$ such that $U=0$ and $X$ is Gaussian. In particular, 
$F(\sigma^2) = h(X)-ch(Y)$ with $X\sim\mathcal{N}(0,a^2)$ for some
$a\le\sigma$. Thus, choosing $P_{UX}$ in \eqref{eq_dstargauss} 
so that $X=U+Z$ with $Z\sim\mathcal{N}(0,a^2)$, 
$U\sim\mathcal{N}(0,\sigma^2-a^2)$, and $U,Z$ independent yields
$\le$ in \eqref{eq_finalequationofthepaper}.
\end{proof}

\small

\bibliographystyle{abbrv}
\bibliography{ref}

\end{document}